\documentclass[12pt]{article}
\usepackage{amssymb}
\usepackage{mathrsfs}
\usepackage{amsmath,amssymb}
\usepackage{amsthm}
\usepackage{graphicx,color}
\usepackage{setspace}
\usepackage{footnote}
\usepackage{algorithm}
\usepackage{algorithmic}
\usepackage{url}
\usepackage{natbib}
\usepackage[top=1in, bottom=1in, left=1in, right=1in]{geometry}

\usepackage{subfig}
\usepackage{soul}
\setstcolor{black}
\usepackage{multirow}
\usepackage{hhline}

\usepackage{multirow}

\setcitestyle{round}
\newtheorem{theorem}{Theorem}
\newtheorem{Prop}{Proposition}
\newtheorem{lemma}{Lemma}
\newtheorem{corollary}{Corollary}

\title{No penalty no tears: Least squares in high-dimensional linear models}
\author{Xiangyu Wang, David Dunson and Chenlei Leng}
\date{First version: January 13, 2015. This version: \today}
\date{}

\begin{document}
\maketitle
\begin{abstract}
Ordinary least squares (OLS) is the default method for fitting linear models, but is
not applicable for problems with dimensionality larger than the sample
size. For these problems, we advocate the use of a generalized version of OLS
motivated by ridge regression, and propose two novel three-step algorithms involving least squares fitting and hard thresholding. The algorithms  are methodologically simple to
understand intuitively, computationally easy to implement efficiently, and
theoretically appealing for choosing models consistently. Numerical exercises comparing  
our methods with penalization-based approaches in simulations and data
analyses illustrate the great potential of the proposed algorithms.
\end{abstract}

\section{Introduction}
Long known for its consistency, simplicity and
optimality under mild conditions, ordinary least squares (OLS) is the most widely used technique for fitting
linear models. Developed originally for fitting fixed dimensional linear models, unfortunately, classical OLS fails in
high dimensional linear models where the number of predictors $p$ far exceeds the
number of observations $n$. To deal with this problem,
\citet{tibshirani1996regression} proposed $\ell_1$-penalized
regression, a.k.a, {\em lasso}, which triggered the recent overwhelming
exploration in both theory and methodology of penalization-based
methods. These methods usually assume that only a small number of coefficients
are nonzero (known as the {\em sparsity} assumption), and minimize the same
least squares loss function as OLS by including an additional penalty on the
coefficients, with the typical choice being the $\ell_1$ norm. Such 
``penalization'' constrains the solution space to certain directions favoring
sparsity of the solution, and thus overcomes the non-unique issue with OLS. It
yields a sparse solution and achieves model selection consistency and estimation
consistency under certain conditions. See \citet{zhao2006model,fan2001variable,zhang2010nearly,zou2005regularization}. 
\par
Despite the success of the methods based on regularization, there are
important issues that can not be easily neglected. On the one hand, methods
using convex penalties, such as {\em lasso}, usually require strong conditions
for model selection consistency \citep{zhao2006model,lounici2008sup}. On the
other hand, methods using non-convex
penalties \citep{fan2001variable,zhang2010nearly} that can achieve model
selection consistency under mild conditions often require huge computational
expense. These concerns have limited the practical use of regularized methods,
motivating alternative strategies such as direct hard thresholding 
\citep{jain2014iterative}. 
\par
In this article, we aim to solve the
problem of fitting high-dimensional sparse linear models by reconsidering OLS
and answering the following simple question: Can ordinary least squares
consistently fit these models with some suitable algorithms? Our result
provides an affirmative answer to this question under fairly general
settings. In particular, we give a generalized form of OLS in high
dimensional linear regression, and develop two algorithms that can consistently
estimate the coefficients and recover the support. These algorithms involve 
least squares type of fitting and hard thresholding, and are 
non-iterative in nature. Extensive empirical
experiments are provided in Section 4 to compare the proposed estimators  to
many existing penalization methods. The performance of the new estimators is very competitive under various setups 
in terms of model selection, parameter estimation and computational time.

\subsection{Related Works}
The work that is most closely related to ours is
\citet{yang2014elementary}, in which the authors proposed an algorithm based on
OLS and ridge regression. However, both their methodology and
theory are still within the $\ell_1$ regularization framework, and their conditions
(especially their C-Ridge and C-OLS conditions) are overly strong and can be
easily violated in practice. \citet{jain2014iterative} proposed an iterative
hard thresholding  algorithm for sparse regression, which shares a similar
spirit of hard thresholding as our algorithm. Nevertheless, their motivation
is completely different, their algorithm lacks theoretical guarantees
for consistent support recovery, and they require an iterative estimation procedure.

\subsection{Our Contributions}
We provide a generalized form of OLS for fitting 
high dimensional data motivated by ridge regression, and develop two
algorithms that can consistently fit linear models on weakly sparse coefficients. We summarize the advantages of our new algorithms in three points. 
\begin{itemize}
\item[1.]
Our algorithms work for highly correlated features under random designs. The consistency of the algorithms relies on
a moderately growing conditional number, as opposed to
the strong irrepresentable condition \citep{zhao2006model,wainwright2009sharp} required by {\em lasso}. 
\item[2.]
Our algorithms can achieve consistent identify strong signals for ultra-high dimensional data ($\log p = o(n)$) with only a bounded variance assumption
on the noise $\varepsilon$, i.e., $var(\varepsilon) < \infty$. This is remarkable as most methods
(c.f. \citet{zhang2010nearly,yang2014elementary,cai2011orthogonal,wainwright2009sharp,zhang2008sparsity,wang2015high}) 
that work for $\log p = o(n)$ case rely on a sub-Gaussian tail/bounded error assumption,
which might fail to hold for real data. \citet{lounici2008sup}
proved that {\em lasso} also achieves consistent model selection with a second-order condition similar to
ours, but requires two additional assumptions. 
\item[3.]
The algorithms are simple, efficient and scale well for large $p$. In particular, 
the matrix operations are fully parallelizable with very few communications for very large $p$,
while regularization methods are either hard to be computed in parallel in the feature space, or 
the parallelization requires a large amount of machine communications.
\end{itemize}

\par
The remainder of this article is organized as follows. In Section 2 we
generalize the ordinary least squares estimator for high dimensional problems
where $p>n$, and propose two three-step algorithms consisting only of least
squares fitting and hard thresholding in a loose sense. Section 3 provides
consistency theory for the algorithms. Section 4 evaluates the empirical
performance. We conclude and discuss further implications of our algorithms in the last section. All the proofs are provided in the supplementary materials.

\section{High dimensional ordinary least squares}
Consider the usual linear model
\begin{align*}
  Y = X\beta + \varepsilon,
\end{align*}
where $X$ is the $n\times p$ design matrix, $Y$ is the $n\times 1$ response
vector and $\beta$ is the coefficient. In the high dimensional
literature, $\beta_i$'s are routinely assumed to be zero except for a small subset
$S^* = supp(\beta)$. In this paper, we consider a slightly more general setting, where
$\beta$ is not exactly sparse, but consists of both strong and weak signals. In particular,
we defined $S^*$ and $S_*$
\begin{align*}
S^* = \{k~: |\beta_k| \geq \tau^*\}\qquad S_* = \{k~: |\beta_k| \leq \tau_*\}
\end{align*}
as the strong and weak signal sets and $S^*\cup S_* = \{1, 2, \cdots, p\}$. The algorithms developed in this paper
is to recover the strong signal set $S^*$. The specific relationship between
$\tau^*$ and $\tau_*$ will be detailed later.
\par
To carefully tailor the low-dimensional OLS estimator in a high dimensional
scenario, one needs to answer the following two questions: i) What is the
correct form of OLS in the high dimensional setting? ii) How to correctly use
this estimator? To answer these, we reconsider OLS from a different
perspective by viewing the OLS as the limit of the ridge estimator with the ridge parameter going to zero, i.e., 
\begin{align*}
  (X^TX)^{-1}X^TY = \lim_{r\rightarrow 0}~ (X^TX + rI_p)^{-1}X^TY.
\end{align*}
One nice property of the ridge estimator is that it exists regardless of the
relationship between $p$ and $n$. A keen observation reveals the following
relationship immediately. 
\begin{lemma}
  For any $p, n, r > 0$, we have
  \begin{align}
    (X^TX + rI_p)^{-1}X^TY = X^T(XX^T + rI_n)^{-1}Y. \label{eq:ridge1}
  \end{align}
\end{lemma}
Notice that the right hand side of \eqref{eq:ridge1} exists when $p>n$ and $r
= 0$. Consequently, we can naturally extend the classical OLS to the high
dimensional scenario by letting $r$ tend to zero in \eqref{eq:ridge1}. Denote
this high dimensional version of the OLS as 
\begin{align*}
  \hat\beta^{(HD)} =  \lim_{r\rightarrow 0} X^T(XX^T + rI_n)^{-1}Y = X^T(XX^T)^{-1}Y.
\end{align*}
The above equation indicates that $ \hat\beta^{(HD)}$ is essentially an
orthogonal projection of $\beta$ onto the row space of $X$. Unfortunately,
this (low dimensional) projection does not have good general performance 
in estimating sparse vectors in high-dimensional cases.
Instead of directly estimating
$\beta$ as $\hat \beta^{HD}$, however, this new estimator of $\beta$ may be used for dimension
reduction by observing $\hat\beta^{(HD)}=X^T(XX^T)^{-1}X\beta+X^T(XX^T)^{-1}\varepsilon = \Phi\beta + \eta$. Since $\eta$ is stochastically small, if $\Phi$ is close to a diagonally dominant matrix and $\beta$ is sparse, 
then the strong and weak signals
can be separated by simply thresholding the small entries of
$\hat\beta^{(HD)}$. The exact meaning of this statement will be discussed in the next section. 
Some simple examples demonstrating the diagonal dominance
of $X^T(XX^T)^{-1}X$ 
are illustrated immediately in Figure
\ref{fig:show}, where the rows of $X$ in the left two plots are drawn from
$N(0, \Sigma)$ with $\sigma_{ij} = 0.6$ or $\sigma_{ij} = 0.99^{|i - j|}$. The sample size and data dimension are chosen as $(n, p) = (50, 1000)$. The
right plot takes the standardized design matrix directly from the real data in Section
4 with $(n, p) = (395, 767)$. A clear diagonal dominance pattern is
visible in each plot.
\begin{figure*}[!tb]
\centering
\includegraphics[height = 4.3cm]{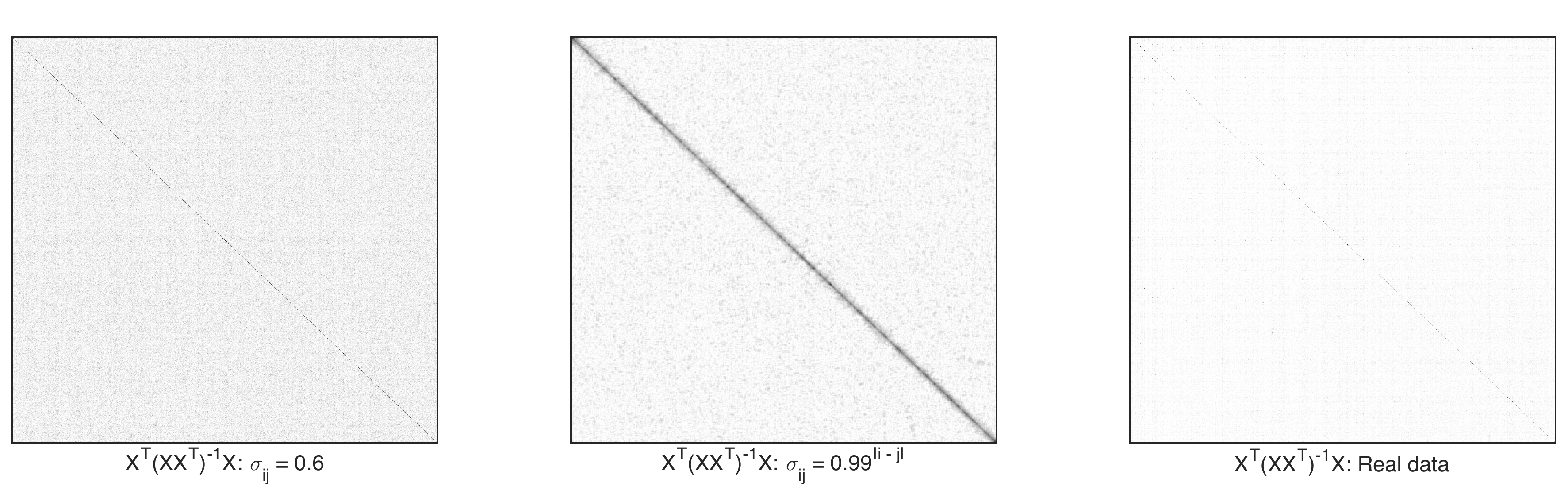}
 \vspace{.3pt}
\caption{Examples for $X^T(XX^T)^{-1}X$. Left: $X\sim N(0, \Sigma)$ with $\sigma_{ij} = 0.6$ and $\sigma_{ii} = 1$; Middle: $X\sim N(0, \Sigma)$ with $\sigma_{ij} = 0.9^{|i-j|}$; Right: Real data from Section 4.}
\label{fig:show}
 \vspace{.3pt}
\end{figure*}
\par
This ability to separate strong and weak signals allows us to first
obtain a smaller model with size $d$ such that $|S^*| < d < n$ containing $S^*$. Since $d$ is below $n$, one can directly apply the usual OLS to
obtain an estimator, which will be thresholded further to obtain a more refined model. The final estimator will then be obtained by an OLS fit on the refined model. This three-stage non-iterative algorithm is termed {\em Least-squares adaptive thresholding (LAT)} and the concrete procedure is described in Algorithm \ref{alg:1}.

\begin{algorithm}[!ht]
  \caption{{\em The Least-squares Adaptive Thresholding (LAT) Algorithm}}
  \label{alg:1}
  \begin{algorithmic}[1]
    \REQUIRE
    \STATE Input $(Y, X), d, \delta$
    \ENSURE \textbf{1 :} Pre-selection
    \STATE Standardize $Y$ and $X$ to $\tilde Y$ and $\tilde X$ having mean 0 and variance 1.
    \STATE Compute $\hat \beta^{(HD)} = \tilde X^T( \tilde X \tilde X^T + 0.1\cdot I_n)^{-1} \tilde Y$.  Rank the importance of the variables by $|\hat \beta^{(HD)}_i|$;
    \STATE Denote the model corresponding to the $d$ largest $|\hat
    \beta^{(HD)}_i|$ as $\mathcal{\tilde M}_d$. Alternatively use extended BIC \citep{chen2008extended} in conjunction with the obtained variable importance to select the best submodel. 
    \ENSURE \textbf{2 :} Hard thresholding
    \STATE $\hat\beta^{(OLS)} = ( X_{\mathcal{\tilde M}_d}^T X_{\mathcal{\tilde M}_d})^{-1} X_{\mathcal{\tilde M}_d}^T Y$;
    \STATE $\hat\sigma^2  = \sum_{i = 1}^n ( y - \hat y)^2/(n - d)$;
    \STATE $\bar C = (X_{\mathcal{\tilde M}_d}^TX_{\mathcal{\tilde M}_d})^{-1}$;
    \STATE Threshold $\hat\beta^{(OLS)}$ by $\textsc{mean}(\sqrt{2\hat
    \sigma^2 \bar C_{ii}\log(4d/\delta)})$ or use BIC to select the best
    submodel. Denote the chosen model as $\mathcal{\hat M}$.
    \ENSURE \textbf{3 :} Refinement
    \STATE $\hat\beta_{\mathcal{\hat M}} = (X_{\mathcal{\hat M}}^TX_{\mathcal{\hat M}})^{-1}X_{\mathcal{\hat M}}^TY$;
    \STATE $\hat\beta_i = 0, \forall i\not\in \mathcal{\hat M}$;\\
    \STATE return $\hat\beta$.
  \end{algorithmic}
\end{algorithm}
\par
The input parameter $d$ is the submodel size selected in Stage 1 and $\delta$ is the tuning parameter determining the threshold in Stage 2; these values will be specified in Section 4. In Stage 1, we use $\hat \beta^{(HD)} = \tilde X^T( \tilde X \tilde X^T + 0.1\cdot I_n)^{-1} \tilde Y$ instead of $\hat \beta^{(HD)} = \tilde X^T( \tilde X \tilde X^T)^{-1} \tilde Y$ because $\tilde X \tilde X^T$ is rank deficient (the rank is $n - 1$) after standardization. The number $0.1$ can be replaced by any arbitrary small number. As noted in \citet{wang2015high}, this additional ridge term is also essential when $p$ and $n$ get closer, in which case the condition number of $\tilde X \tilde X^T$ increases dramatically, resulting in an explosion of the model noise. Our results in Section 3 mainly focus on $\hat \beta^{(HD)} = X^T(XX^T)^{-1}Y$ where $X$ is assumed to be drawn from a distribution with mean zero, so no standardization or ridge adjustment is required. However, the result is easy to generalize to the case where a ridge term is included. See \citet{wang2015high}.
\par
The Stage 1 of Algorithm \ref{alg:1} is very similar to variable screening methods \citep{fan2008sure,wang2015high}. However, most screening methods require a sub-Gaussian condition the noise to handle the ultra-high dimensional data where $\log(p) = o(n)$. In contrast to the existing theory, we prove in the next section a better result that Stage 1 of Algorithm \ref{alg:1} can produce satisfactory submodel even for heavy-tailed noise.
\par
The estimator $\hat\beta^{(OLS)}$ in Stage 2 can be substituted by its ridge counterpart $\hat\beta^{(Ridge)} = (X_{\mathcal{\tilde M}_d}^TX_{\mathcal{\tilde
    M}_d} + rI_d)^{-1} X_{\mathcal{\tilde M}_d}^T Y$ and $\bar C$ by $(X_{\mathcal{\tilde M}_d}^TX_{\mathcal{\tilde
    M}_d} + rI_d)^{-1}$ to stabilize numerical computation. Similar modification can be applied to the Stage 3 as well. The resulted variant of
the algorithm is referred to as the {\em Ridge Adaptive Thresholding (RAT)} algorithm and described in Algorithm \ref{alg:2}.

\begin{algorithm}[!ht]
  \caption{{\em The Ridge Adaptive Thresholding (RAT) Algorithm}}
  \label{alg:2}
  \begin{algorithmic}[1]
    \REQUIRE
    \STATE Input $(Y, X), d, \delta, r$
    \ENSURE \textbf{1 :} Pre-selection
    \STATE Standardize $Y$ and $X$ to $\tilde Y$ and $\tilde X$ having mean 0 and variance 1.
    \STATE Compute $\hat \beta^{(HD)} = \tilde X^T( \tilde X \tilde X^T + 0.1\cdot I_n)^{-1} \tilde Y$.  Rank the importance of the variables by $|\hat \beta^{(HD)}_i|$;
    \STATE Denote the model corresponding to the $d$ largest $|\hat
    \beta^{(HD)}_i|$ as $\mathcal{\tilde M}_d$. Alternatively use eBIC in \citet{chen2008extended} in conjunction with the obtained variable importance to select the best    submodel. 
    \ENSURE \textbf{2 :} Hard thresholding
    \STATE $\hat\beta^{(Ridge)} = ( X_{\mathcal{\tilde M}_d}^T X_{\mathcal{\tilde M}_d} + rI_d)^{-1} X_{\mathcal{\tilde M}_d}^T Y$;
    \STATE $\hat\sigma^2  = \sum_{i = 1}^n ( y - \hat y)^2/(n - d)$;
    \STATE $\bar C = (X_{\mathcal{\tilde M}_d}^TX_{\mathcal{\tilde M}_d} + rI_d)^{-1}$;
    \STATE Threshold $\hat\beta^{(OLS)}$ by $\textsc{mean}(\sqrt{2\hat
    \sigma^2 \bar C_{ii}\log(4d/\delta)})$ or use BIC to select the best
    submodel. Denote the chosen model as $\mathcal{\hat M}$.
    \ENSURE \textbf{3 :} Refinement
    \STATE $\hat\beta_{\mathcal{\hat M}} = (X_{\mathcal{\hat M}}^TX_{\mathcal{\hat M}} + rI)^{-1}X_{\mathcal{\hat M}}^TY$;
    \STATE $\hat\beta_i = 0, \forall i\not\in \mathcal{\hat M}$;\\
    \STATE return $\hat\beta$.
  \end{algorithmic}
\end{algorithm}
We suggest to use 10-fold cross-validation to tune the ridge parameter $r$. Notice that the model is already small after stage 1, so using cross-validation will not significantly increase the computational burden. The computational performance is illustrated in Section 4.

\section{Theory}
In this section, we prove the consistency of Algorithm \ref{alg:1} in identifying strong signals and provide concrete forms for all the values needed for the algorithm to work. Recall the linear model $Y = X\beta + \varepsilon$. We consider the random design where the rows of $X$ are drawn from an elliptical distribution with a density of $g(x_i^T\Sigma^{-1} x_i)$ for some nonnegative function $g$ and positive definite $\Sigma$. It is easy to show that $x_i$ admits an equivalent representation as
\begin{align}
 x_i \stackrel{(d)}{=} L_i \frac{\sqrt{p}z_i}{\|z_i\|_2} \Sigma^{1/2} = \frac{\sqrt{p}L_i}{\|z_i\|_2} z_i\Sigma^{1/2}. \label{eq:ep}
\end{align} 
where $z_i$ is a p-variate standard Gaussian random variable and $L_i$ is a nonnegative random variable that is independent of $z_i$. We denote this distribution by $EN(L, \Sigma)$.
This random design allows for various correlation structures among predictors and contains many distribution families that are widely used to illustrate methods that rely on the restricted eigenvalue conditions \citep{bickel2009simultaneous,raskutti2010restricted}. The noise $\varepsilon$, as mentioned earlier, is only assumed to have the second-order moment, i.e., $var(\varepsilon) =\sigma^2 < \infty$, in contrast to the sub-Gaussian/bounded error assumption seen in most high dimension literature. See \citet{zhang2010nearly,yang2014elementary,cai2011orthogonal,wainwright2009sharp,zhang2008sparsity}. This relaxation is similar to \citet{lounici2008sup}; however we do not require any further assumptions needed by \citet{lounici2008sup}. In Algorithm \ref{alg:1}, we also propose to use extended BIC and BIC for parameter tuning. However, the corresponding details will not be pursued here, as their consistency is straightforwardly implied by the results from this section and the existing literature on extended BIC and BIC \citep{chen2008extended}.

As shown in \eqref{eq:ep}, the variable $L$ controls the signal strength of $x_i$, we thus need a lower bound on $L_i$ to guarantee a good signal strength. Define $\kappa = cond(\Sigma)$. We state our result in three theorems.
\begin{theorem}\label{thm:1}
Assume $x_i\sim EN(L_i, \Sigma)$ with $E[L_i^{-2}] < M_1$ and $\varepsilon_i$ is a random variable with a bounded variance $\sigma^2$. We also assume $p>c_0n$ for some $c_0>1$ and $var(Y) \leq M_0$. If $ |S^*|\log p = o(n)$, $n>4c_0/(c_0 - 1)^2$, and $\tau^*/\tau_* \geq 4\kappa^2$, then we can choose $\gamma$ to be $\frac{2c_1\kappa^{-1}\tau}{3}\frac{n}{p}$, where $c_1$ is some absolute constant specified in Lemma \ref{lemma:5} and for any $\alpha \in (0, 1)$ we have
{\small
\begin{align*}
P\bigg(\max_{i\in S_*}|\hat\beta_i^{(HD)}| &\leq \gamma \leq \min_{i\in S^*} |\hat\beta_i^{(HD)}|\bigg)= 1 -
O\bigg(\frac{\sigma^2\kappa^4\log p}{\tau^{*2} n^{\alpha}}\bigg).
\end{align*}}
\end{theorem}
Theorem \ref{thm:1} guarantees the model selection consistency of the first stage of Algorithm
\ref{alg:1}. It only requires a second-moment condition on the noise tail, relaxing the sub-Gaussian assumption seen in other literature. The probability term shows that the algorithm requires the important signals to be lower bounded by a signal strength of $\sigma\sqrt{\frac{\log p}{n^{\alpha}}}$ with a positive $\alpha$. In addition, a gap of $\tau^*/\tau_*\geq 4\kappa^2$ is needed between the strong signals and the weak signals in order for a successful support recovery.
\par
 As $\gamma$ is not easily computable based on
data, we propose to rank the $|\hat\beta_i^{(HD)}|'s$ and select $d$ largest
coefficients. Alternatively, we can construct a series of nested models formed
by ranking the largest $n$ coefficients and adopt the extended BIC 
\citep{chen2008extended} to select the best submodel. Once the submodel $\mathcal{\tilde M}_d$ is obtained, we proceed to the second stage by obtaining an estimate via ordinary least squares $\hat\beta^{(OLS)}$ corresponding to $\mathcal{\tilde M}_d$. The theory for $\hat\beta^{(OLS)}$ requires more stringent conditions, as we now need to estimate $\beta_{\mathcal{\hat M}_d}$ instead of just a correct ranking. In particular, we have to impose conditions on the magnitude of $\beta_{S_*}$ and the moments of $L$, i.e., for $\hat\beta^{(OLS)}$ we have the following result.

\begin{theorem}\label{thm:2}
Assume the same conditions for $X$ and $\varepsilon$ as in Theorem \ref{thm:1}. We also assume $n \geq 64\kappa d\log p$ and $d - |S^*| \leq \tilde c$ for some $\tilde c > 0$. If $E[L^{-12}] \leq M_1$, $E[L^{12}] \leq M_2$, $\tau_*\leq \frac{\sigma}{\kappa}\sqrt{\frac{\log p}{n}}$ and there exists some $\iota\in (0, 1)$ such that $\sum_{i\in S_*}|\beta_i|^\iota\leq R$, then for any $\alpha > 0$, we have
\begin{align*}
P\bigg(\max_{|\mathcal{\hat M}| \leq d, ~S^*\subset\mathcal{\hat M}} \|\hat\beta^{(OLS)} - \beta\|_\infty &\leq 2\sigma\sqrt{\frac{\log p}{n^\alpha}}\bigg) \\
&= 1 - O\bigg(\frac{\lambda_*^{-2} d\log d}{n^{\frac{1}{3}(1 - \alpha)}} + \frac{M_1 + M_2}{n^{\frac{1}{3}(1 - 4\alpha)}} + \frac{(M_1 + M_2)R^3}{(\log p)^{2\iota}n^{3 - 4\alpha - 2\iota}}\bigg),
\end{align*}
i.e., if $\tau^* \geq 5\sigma\sqrt{\frac{\log p}{n^\alpha}}$, then we can choose $\gamma' = 3\sigma\sqrt{\frac{\log p}{n^\alpha}}$ and
\begin{align*}
\max_{i\not\in S^*}|\hat\beta_i^{(OLS)}| &\leq \gamma' \leq \min_{i\in S^*} |\hat\beta_i^{(OLS)}|
\end{align*}
with probability tending to 1.
\end{theorem}
The moment condition on $L$ is not tight. We used this number just for simplicity. As shown in Theorem \ref{thm:2}, the $l_\iota$ norm of $\beta_{S_*}$ is allowed to grow in the rate of $(\log p)^{2\iota/3} n^{1 - 4\alpha/3 - 2\iota/3}$, i.e., our algorithm works for weakly sparse coefficients. However, Theorem \ref{thm:2} imposes an upper bound on $\alpha$ while Theorem \ref{thm:1} not. This is mainly due to the moment assumption on $L$ and the different structure between $\hat\beta^{(HD)}$ and $\hat\beta^{(OLS)}$, i.e., $\hat\beta^{(HD)}$ does not rely on $L$ for diminishing the unimportant signals. For ridge regression, we have the following result. 
\begin{theorem}[Ridge regression] \label{thm:3}
Assume all the conditions in Theorem \ref{thm:2}. If we choose the ridge parameter satisfying
\begin{align*}
r \leq \frac{\sigma n^{(7/9 - 5\alpha/18)}\sqrt{\log p}}{162\kappa M_0},
\end{align*}
then we have
\begin{align*}
P\bigg(\max_{|\mathcal{\hat M}| \leq d, S^*\subset\mathcal{\hat M}} \|\hat\beta^{(ridge)} - \beta\|_\infty &\leq 3\sigma\sqrt{\frac{\log p}{n^\alpha}}\bigg) \\
&= 1 - O\bigg(\frac{\lambda_*^{-2} d\log d}{n^{\frac{1}{3}(1 - \alpha)}} + \frac{2M_1 + M_2}{n^{\frac{1}{3}(1 - 4\alpha)}} + \frac{(M_1 + M_2)R^3}{(\log p)^{2\iota}n^{3 - 4\alpha - 2\iota}}\bigg),
\end{align*}
i.e., if $\tau^* \geq 7\sigma\sqrt{\frac{\log p}{n^\alpha}}$, then we can choose $\gamma' = 4\sigma\sqrt{\frac{\log p}{n^\alpha}}$ and
\begin{align*}
\max_{i\not\in S^*}|\hat\beta_i^{(Ridge)}(r)| &\leq \gamma' \leq \min_{i\in S^*}|\hat\beta_i^{(Ridge)}(r)|
\end{align*}
with probability tending to 1.
\end{theorem}

Note that the ridge parameter $r$ can be chosen as a constant, bypassing the need to specify $r$ at least in theory. 
When both the noise $\varepsilon$ and $X$ follows Gaussian distribution and $\tau_* = 0$, we can obtain a more explicit form of the threshold $\gamma'$, as the following Corollary shows.
\begin{corollary}[Gaussian noise] \label{cor:1}
Assume $\varepsilon\sim N(0, \sigma^2)$, $X\sim N(0, \Sigma)$ and $\tau_* = 0$. For any $\delta \in (0, 1)$, define $\gamma' = 8\sqrt{2}\hat \sigma \sqrt{\frac{2\kappa\log(4d/\delta)}{n}}$, where $\hat\sigma$ is the estimated standard error as $\hat\sigma^2 = \sum_{i = 1}^n (y_i - \hat y_i)^2/(n - d)$. For sufficiently large $n$, if $d \leq n - 4K^2\log(2/\delta)/c$ for some absolute constants $c$, $K$ and $\tau^* \geq 24\sigma \sqrt{\frac{2\kappa\log(4d/\delta)}{n}}$, then with probability at least $1 - 2\delta$, we have
\begin{align*}
|\hat\beta_i^{(OLS)}|\geq\gamma' \quad \forall i\in S^*\quad\mbox{ and }\quad |\hat\beta_i^{(OLS)}|\leq\gamma' \quad\forall i\not\in S_*.
\end{align*}
\end{corollary}

Write $\bar C = (X_{\mathcal{\tilde M}_d}^TX_{\mathcal{\tilde M}_d})^{-1}$ as in Algorithm 1. In practice, we propose to use $\gamma' = mean(\sqrt{2\hat \sigma^2\bar C_{ii}\log(4d/\delta)})$ as the threshold (see Algorithm \ref{alg:1}), because the estimation error takes a form of $\sqrt{\sigma^2\bar C_{ii}\log(4d/\delta)}$. Alternatively, instead of identifying an explicit form of the threshold value (as is hard for general noise), one may also use BIC on nested models formed by ranking $|\hat\beta^{(OLS)}|$ to search for the true model. Once the final model is obtained, as in Stage 3 of Algorithm 1, we refit it again using ordinary least squares. The final output will have the same output as if we knew the true model {\em a priori} with probability tending to 1, i.e., we have the following result.
\begin{theorem}\label{thm:4}
Let $\mathcal{\hat M}$ and $\hat\beta$ be the final output from LAT or RAT. Assume all conditions in Theorem \ref{thm:1}, \ref{thm:2} and \ref{thm:3}. Then with probability at least $1 - O\bigg(\frac{\lambda_*^{-2} d\log d}{n^{\frac{1}{3}(1 - \alpha)}} + \frac{M_1 + M_2}{n^{\frac{1}{3}(1 - 4\alpha)}} + \frac{(M_1 + M_2)R^3}{(\log p)^{2\iota}n^{3 - 4\alpha - 2\iota}}\bigg)$ we have
\begin{align*}
\mathcal{\hat M} = S^*,~\|\hat\beta_{S^*} - \beta_{S^*}\|_2^2\leq \frac{2|S^*|\sigma^2\log p}{n^\alpha},~\mbox{and}~ \|\hat\beta - \beta\|_\infty\leq 2\sigma\sqrt{\frac{\log p }{n^{\alpha}}}.
\end{align*}
\end{theorem}
As implied by Theorem \ref{thm:1} -- \ref{thm:4}, {\em LAT} and {\em RAT} can consistently identify strong signals in the ultra-high dimensional ($\log p = o(n)$) setting with only the bounded moment assumption $var(\varepsilon) < \infty$,  in contrast to most existing methods that require $\varepsilon \sim N(0, \sigma^2)$ or $\|\varepsilon\|_\infty < \infty$.

\section{Experiments}
In this section, we provide extensive numerical experiments for assessing the performance of {\em LAT} and {\em RAT}. In particular, we compare the two methods to existing penalized methods including {\em lasso}, elastic net ({\em enet} \citep{zou2005regularization}), {\em adaptive lasso} \citep{zou2006adaptive}, {\em scad} \citep{fan2001variable} and {\em mc+} \citep{zhang2010nearly}. As it is well-known that the {\em lasso} estimator is biased, we also consider two variations of it by combining {\em lasso} with Stage 2 and 3 of our {\em LAT} and {\em RAT} algorithms, denoted as {\em lasLAT} ({\em las1} in Figures) and {\em lasRAT} ({\em las2} in Figures) respectively. We note that the {\em lasLat} algorithm is very similar to the thresholded lasso \citep{zhou2010thresholded} with an additional thresholding step. We code {\em LAT} and {\em RAT} and {\em adaptive lasso} in {\em Matlab}, use \texttt{glmnet} \citep{friedman2010regularization} for {\em enet} and {\em lasso}, and \texttt{SparseReg} \citep{zhou2012path,zhou2013path} for {\em scad} and {\em mc+}. Since {\em adaptive lasso} achieves a similar performance as {\em lasLat} on synthetic datasets, we only report its performance for the real data.

\subsection{Synthetic Datasets}
The model used in this section for comparison is the linear model $Y = X\beta + \varepsilon$, where $\varepsilon\sim N(0, \sigma^2)$ and $X\sim N(0, \Sigma)$. To control the signal-to-noise ratio, we define $r = \|\beta\|_2/\sigma$, which is chosen to be $2.3$ for all experiments. The sample size and the data dimension are chosen to be $(n, p) = (200, 1000)$ or $(n, p) = (500, 10000)$ for all experiments. For evaluation purposes, we consider four different structures of $\Sigma$ below.
\par
\noindent \textbf{(i) \emph{Independent predictors}}.
The support is set as $S=\{1,2,3,4,5\}$. We generate $X_i$ from a standard multivariate
normal distribution with independent components. The coefficients are specified as
\[
\beta_i = \left\{
\begin{array}{ll}
(-1)^{u_i}(|N(0,1)|+1), ~u_i\sim Ber(0.5) & i\in S\\
0 & i\not\in S.
\end{array}
\right.
\]
\par

\noindent \textbf{(ii) \emph{Compound symmetry}}. 
All predictors are equally correlated with correlation
$\rho = 0.6$. The coefficients are set to be $\beta_i = 3$ for $i=1,...,5$ and $\beta_i
= 0$ otherwise.

\noindent \textbf{(iii) \emph{Group structure}}. 
This example is  Example 4 in \citet{zou2005regularization}, for which we
allocate the 15 true variables into three groups. Specifically, the predictors are generated as
\begin{align*}
x_{1+3m} &= z_1+N(0, 0.01),\\
x_{2+3m} &= z_2+N(0, 0.01),\\
x_{3+3m} &= z_3+N(0, 0.01),
\end{align*}
where $m=0,1,2,3,4$ and $z_i\sim N(0,1)$ are independent. The coefficients are set as
\begin{align*}
\beta_i = 3, ~i=1,2,\cdots,15;~ \beta_i = 0, ~i=16,\cdots,p.
\end{align*}

\noindent \textbf{(iv) \emph{Factor models}}. 
This model is also considered in 
\citet{meinshausen2010stability} and 
\citet{cho2012high}. Let $\phi_j,
j=1,2,\cdots,k$ be independent standard normal variables. We set
predictors as $x_i = \sum_{j=1}^k \phi_j f_{ij}+\eta_i$, where $f_{ij}$ and
$\eta_i$ are generated from independent standard normal distributions. The
number of factors is chosen as $k=5$ in the simulation while
the coefficients are specified the same as in Example (ii).

\begin{figure*}[!htbp]
\centering
\includegraphics[height = 3.8cm]{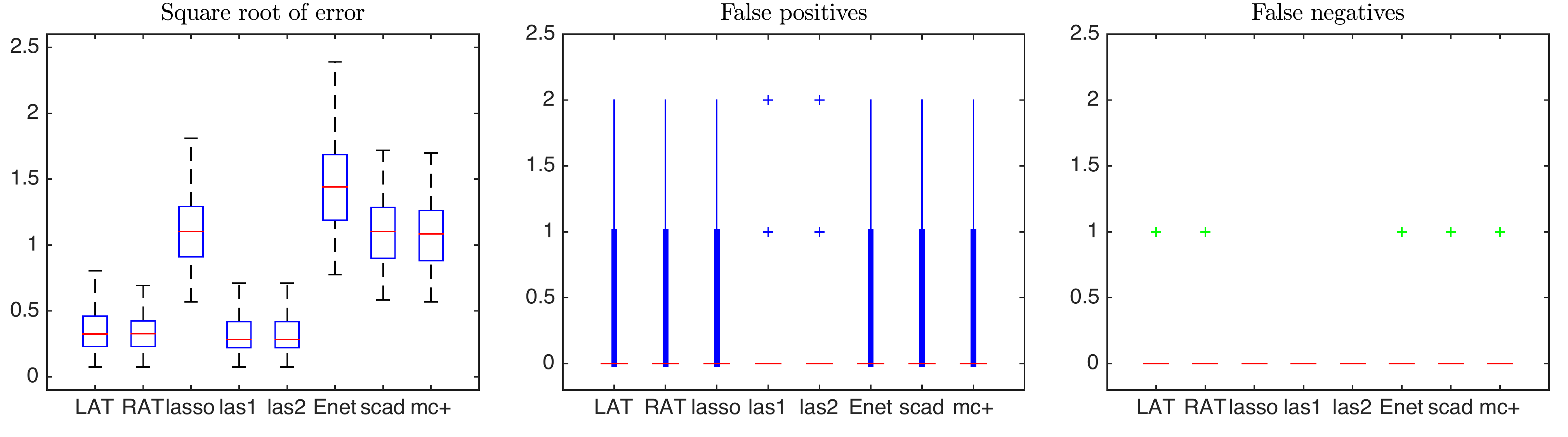}
 \vspace{.3in}
\caption{The Boxplots for Example (i). Left: Estimation Error; Middle: False Positives; Right: False Negatives}
\label{fig:1}
\end{figure*}

\begin{figure*}[!htpb]
\centering
\includegraphics[height = 3.8cm]{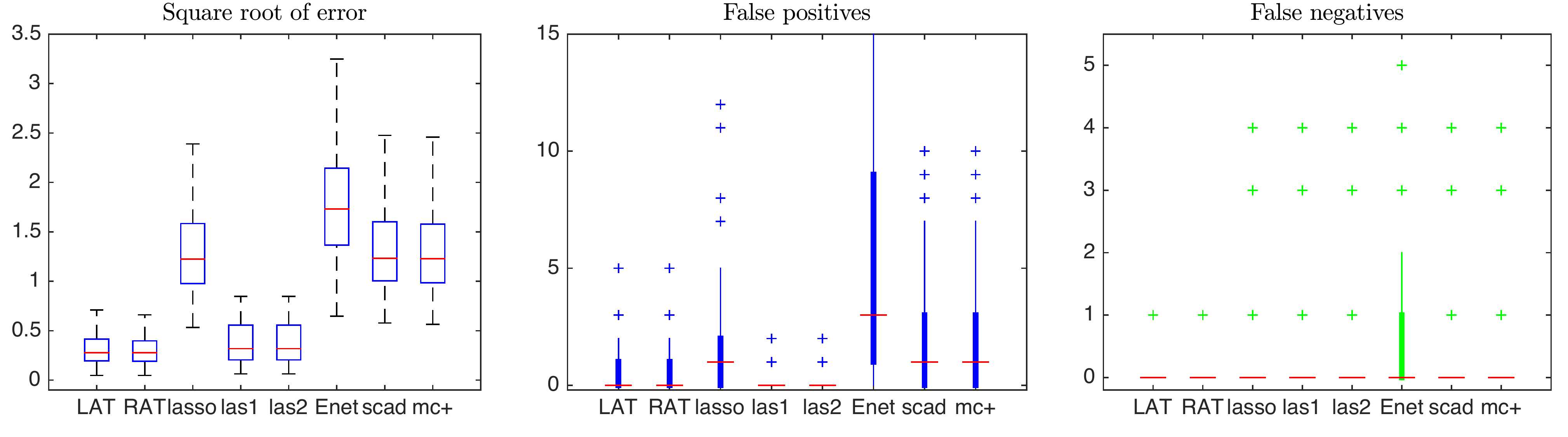}
 \vspace{.3in}
\caption{The Boxplots for Example (ii). Left: Estimation Error; Middle: False Positives; Right: False Negatives}
\label{fig:2}
\end{figure*}

\begin{figure*}[!htbp]
\centering
\includegraphics[height = 3.8cm]{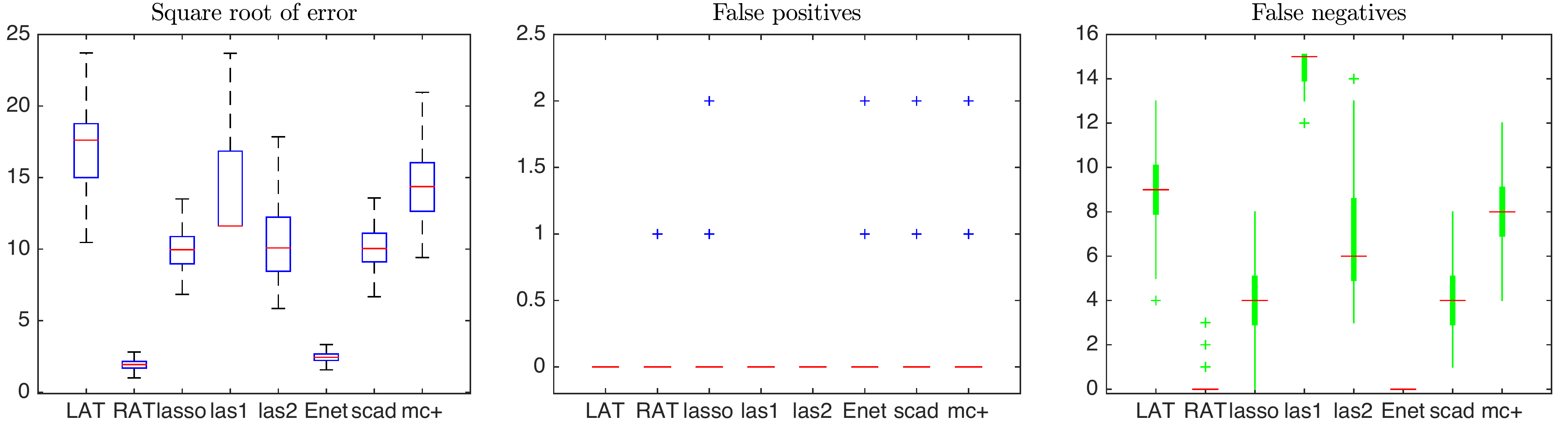}
 \vspace{.3in}
\caption{The Boxplots for Example (iii). Left: Estimation Error; Middle: False Positives; Right: False Negatives}
\label{fig:3}
\end{figure*}

\begin{figure*}[!htbp]
\centering
\includegraphics[height = 3.8cm]{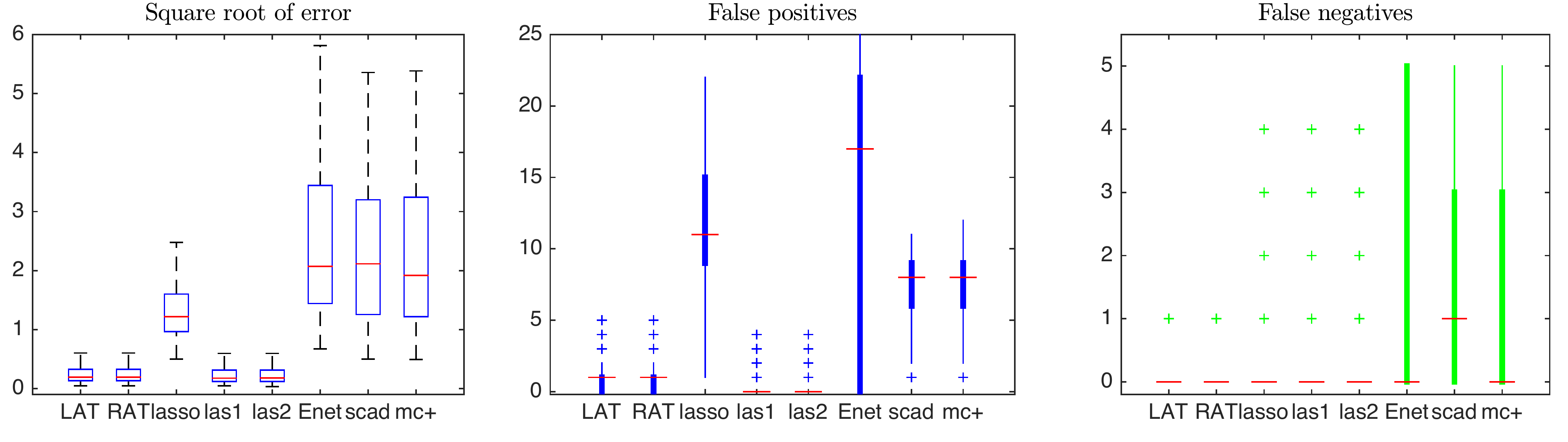}
 \vspace{.3in}
\caption{The boxplots for Example (iv). Left: Estimation Error; Middle: False Positives; Right: False Negatives}
\label{fig:4}
\end{figure*}

\begin{table*}[!tb]
\caption{Results for $(n, p) = (500, 10000)$}
\label{tab:2}
\begin{center}
\begin{tabular}{l|l|p{1.3cm}p{1.3cm}p{1.3cm}p{1.3cm}p{1.3cm}p{1.3cm}p{1.3cm}p{1.3cm}}
\hline
Example       &          & {\em LAT} & {\em RAT} & {\em lasso} & {\em lasLAT} & {\em lasRAT} & {\em enet} & {\em scad} & {\em mc+}\\
\hline\hline
           & RMSE     &  0.263    & 0.264     & 0.781       & 0.214        & 0.214        & 1.039      & 0.762      & 0.755\\
Ex.(i)     & \# FPs   &  0.550    & 0.580     & 0.190       & 0.190        & 0.190        & 0.470      & 0.280      & 0.280\\
           & \# FNs   &  0.010    & 0.010     & 0.000       & 0.000        & 0.000        & 0.000      & 0.000      & 0.000\\
           & Time  &  36.1     & 41.8      & 72.7        & 72.7         & 74.1         & 71.8       & 1107.5     & 1003.2\\
\hline
           & RMSE     &  0.204    & 0.204     & 0.979       & 0.260        & 0.260        & 1.363      &0.967       & 0.959\\
Ex. (ii)   & \# FPs   &  0.480    & 0.480     & 1.500       & 0.350        & 0.350        & 10.820     & 2.470      & 2.400\\
           & \# FNs   &  0.000    & 0.000     & 0.040       & 0.040        & 0.040        & 0.040    & 0.020      & 0.020\\
           & Time  &  34.8     & 40.8      & 76.1        & 76.1         & 77.5         & 82.0       & 1557.6     & 1456.1\\
\hline
           & RMSE     &  9.738    & 1.347     & 7.326       & 17.621       & 3.837        & 1.843      & 7.285      & 8.462\\
Ex. (iii)  & \# FPs   &  0.000    & 0.000     & 0.060       & 0.000        & 0.000        & 0.120      & 0.120      & 0.090\\
           & \# FNs   &  4.640    & 0.000     & 1.440       & 13.360       & 1.450        & 0.000      & 1.800      & 2.780\\
           & Time  &  35.0     & 41.6      & 75.6        & 75.6         & 77.5         & 74.4       & 6304.4     & 4613.8\\
\hline
           & RMSE     &  0.168    & 0.168     & 1.175       & 0.256        & 0.256        & 1.780      & 0.389      & 0.368\\
Ex. (iv)   & \# FPs   &  0.920    & 0.920     & 21.710      & 0.260        & 0.260        & 37.210     & 6.360      & 6.270\\
           & \# FNs   &  0.010    & 0.010     & 0.140       & 0.140        & 0.140        & 0.450      & 0.000      & 0.000\\
           & Time  &  34.5     & 41.1      & 78.7        & 78.7         & 80.8         & 81.4       & 1895.6     & 1937.1\\
\hline
\hline
\end{tabular}
\end{center}
\end{table*}

To compare the performance of all methods, we simulate $200$ synthetic datasets for $(n, p) = (200, 1000)$ and $100$ for $(n, p) = (500, 10000)$ for each example, and record i) the \textbf{root mean squared error (RMSE):} $\|\hat\beta - \beta\|_2$, ii) the \textbf{false negatives (\# FN)}, iii) the \textbf{false positives (\# FP)} and iv) the actual \textbf{runtime} (in milliseconds). We use the extended BIC \citep{chen2008extended} to choose the parameters for any regularized algorithm. Due to the huge computation expense for {\em scad} and {\em mc+}, we only find the first $\lceil\sqrt{p}\rceil$ predictors on the solution path (because we know $s\ll\sqrt{p}$). For {\em RAT} and {\em LAT}, $d$ is set to $0.3\times n$. For {\em RAT} and {\em larsRidge}, we adopt a 10-fold cross-validation procedure to tune the ridge parameter $r$ for a better finite-sample performance, although the theory allows $r$ to be fixed as a constant. For all hard-thresholding steps, we fix $\delta = 0.5$. The results for $(n, p) = (200, 1000)$ are plotted in Figure \ref{fig:1}, \ref{fig:2}, \ref{fig:3} and \ref{fig:4} and a more comprehensive result (average values for {\bf RMSE, \# FPs, \# FNs, runtime}) for $(n, p) = (500, 10000) $ is summarized in Table \ref{tab:2}. 

As can be seen from both the plots and the tables, {\em LAT} and {\em RAT} achieve the smallest RMSE for Example (ii), (iii) and (iv) and are on par with {\em lasLAT} for Example (i). For Example (iii), {\em RAT} and {\em enet} achieve the best performance while all the other methods fail to work. In addition, the runtime of {\em LAT} and {\em RAT} are also competitive compared to that of {\em lasso} and {\em enet}. We thus conclude that {\em LAT} and {\em RAT} achieve similar or even better performance compared to the usual regularized methods.

\subsection{A Student Performance Dataset}
We look at one dataset used for evaluating student achievement in Portuguese schools \citep{cortez2008using}. The data attributes include student grades and school related features that were collected by using school reports and questionnaires. The particular dataset used here provides the students' performance in mathematics. The goal of the research is to predict the final grade based on all the attributes. 

\par
The original data set contains 395 students and 32 raw attributes. The raw attributes are recoded as 40 attributes and form 780 features after interaction terms are added. We then remove features that are constant for all students. This gives 767 features for us to work with. To compare the performance of all methods, we first randomly split the dataset into 10 parts. We use one of the 10 parts as a test set, fit all the methods on the other 9 parts, and then record their prediction error (root mean square error, RMSE), model size and runtime on the test set. We repeat this procedure until each of the 10 parts has been used for testing. The averaged prediction error, model size and runtime are summarized in Table \ref{tab:real}. We also report the performance of the null model which predicts the final grade on the test set using the mean final grade in the training set.

\begin{table*}[!tb]
\centering
\caption{Prediction Error of the Final Grades by Different Methods}
\begin{tabular}{p{2.9cm}|p{3cm}p{3cm}p{3cm}p{3cm}}
\hline
methods     & mean error      & Standard error  & average model size & runtime (millisec)\\
\hline\hline
{\em LAT}   & 1.93        & 0.118             & 6.8                      & 22.3\\
{\em RAT}   & \textbf{1.90}        & 0.131             & 3.5                      & 74.3\\
{\em lasso} & 1.94        & 0.138             & 3.7                    & 60.7\\
{\em lasLAT}& 2.02        & 0.119             & 3.6                     & 55.5\\
{\em lasRAT}& 2.04        & 0.124             & 3.6                      & 71.3\\
{\em enet}  & 1.99        & 0.127             & 4.7                    & 58.5\\
{\em scad}  & 1.92        & 0.142             & 3.5                      & 260.6\\
{\em mc+}   & 1.92        & 0.143             & 3.4                     & 246.0\\
{\em adaptive lasso}& 2.01        & 0.140             & 3.6                     & 65.5\\
{\em null}  & 4.54        & 0.151             & 0                      & ---\\  
\hline\hline
\end{tabular}
\label{tab:real}
\end{table*}

\begin{table*}[!tb]
\centering
\caption{Prediction Error of the Final Grades Excluding Strong Signals}
\begin{tabular}{p{2.9cm}|p{3cm}p{3cm}p{3cm}p{3cm}}
\hline
methods     & mean error      & Standard error  & average model size & runtime (millisec)\\
\hline\hline
{\em LAT}   & 4.50        & 0.141             & 5.3                      & 22.4\\
{\em RAT}   & 4.26        & 0.130             & 4.0                      & 74.0\\
{\em lasso} & 4.27        & 0.151             & 5.0                    & 318.9\\
{\em lasLAT}& 4.25        & 0.131             & 2.9                     & 316.5\\
{\em lasRAT}& 4.28        & 0.127             & 2.8                      & 331.9\\
{\em enet}  & 4.37        & 0.171             & 6.0                    & 265.6\\
{\em scad}  & 4.30        & 0.156             & 4.8                      & 387.5\\
{\em mc+}   & 4.29        & 0.156             & 4.7                     & 340.2\\
{\em adaptive lasso}& \textbf{4.24}        & 0.180             & 4.8                     & 298.0\\
{\em null}  & 4.54        & 0.151             & 0                      & ---\\  
\hline\hline
\end{tabular}
\label{tab:notreal}
\end{table*}

It can be seen that {\em RAT} achieves the smallest cross-validation error, followed by {\em scad} and {\em mc+}. In the post-feature-selection analysis, we found that two features, the 1st and 2nd period grades of a student, were selected by all the methods. This result coincides with the common perception that these two grades usually have high impact on the final grade. 
\par
In addition, we may also be interested in what happens when no strong signals are presented. One way to do this is to remove all the features that are related to the 1st and 2nd grades before applying the aforementioned procedures. The new result without the strong signals removed are summarized in Table \ref{tab:notreal}.

Table \ref{tab:notreal} shows a few interesting findings. First, under this artificial weak signal scenario, {\em adaptive lasso} achieves the smallest cross-validation error and {\em RAT} is the first runner-up. Second, in Stage 1, {\em lasso} seems to provide slightly more robust screening than OLS in that the selected features are less correlated. This might be the reason that {\em LAT} is outperformed by {\em lasLAT}. However, in both the strong and weak signal cases, {\em RAT} is consistently competitive in terms of performance.

\section{Conclusion}
We have proposed two novel algorithms {\em Lat} and {\em Rat} that only rely on least-squares type of fitting and hard thresholding, based on a high-dimensional generalization of OLS. The two methods are simple, easily implementable, and can consistently fit a high dimensional linear model and recover its support. The performance of the two methods are competitive compared to existing regularization methods. It is of great interest to further extend this framework to other models such as generalized linear models and models for survival analysis.

\bibliography{LAT-arXiV}

\newpage

\section*{Appendix 0: Proof of Lemma 1}
Applying the Sherman-Morrison-Woodbury formula
\begin{equation*}
(A+UDV)^{-1} = A^{-1} - A^{-1}U(D^{-1}+VA^{-1}U)^{-1}VA^{-1},
\end{equation*}
we have
\begin{align*}
r(r I_p+X^TX)^{-1} = I_p - X^T(I_n+\frac{1}{r}XX^T)^{-1}X\frac{1}{r} = I_p-X^T(r I_n+XX^T)^{-1}X.
\end{align*}
Multiplying $X^TY$ on both sides, we get
\begin{equation*}
r(r I_p+X^TX)^{-1}X^TY = X^TY-X^T(r I_n+XX^T)^{-1}XX^TY. 
\end{equation*}
The right hand side can be further simplified as
\begin{align*}
\notag X^TY&-X^T(r I_n+XX^T)^{-1}XX^TY \\
\notag &= X^TY-X^T(r I_n+XX^T)^{-1}(r I_n +XX^T - r I_n)Y\\
 &= X^TY-X^TY+r(r I_n+XX^T)^{-1}Y = r X^T(r I_n+XX^T)^{-1}Y. 
\end{align*}
Therefore, we have
\[ (r I_p+X^TX)^{-1}X^TY = X^T(r I_n+XX^T)^{-1}Y.\]

\section*{Appendix A: Proof of Theorem \ref{thm:1}}
Recall the estimator $\hat \beta^{(HD)} = X^T(XX^T)^{-1}Y = X^T(XX^T)^{-1}X\beta + X^T(XX^T)^{-1}\varepsilon = \xi + \eta$. The following three lemmas will be used to bound $\xi$ and $\eta$ respectively.
\begin{lemma}\label{lemma:5}
  Let $\Phi = X^T(XX^T)^{-1}X$. Assume $p>c_0n$ for some $c_0>1$, then for any $C>0$ there exists some $0<c_1<1<c_2$ and $c_3>0$ such that for any $t>0$ and any $i\in Q, j\neq i$,
\begin{align}
  P\bigg(|\Phi_{ii}|<c_1\kappa^{-1}\frac{n}{p})\leq 2e^{-Cn},\quad |\Phi_{ii}|>c_2\kappa\frac{n}{p}\bigg)\leq 2e^{-Cn}
\end{align}
and
\begin{align}
  P\bigg(|\Phi_{ij}|>c_4\kappa t\frac{\sqrt{n}}{p}\bigg) \leq 5e^{-Cn} + 2e^{-t^2/2},
\end{align}
where $c_4 = \frac{\sqrt{c_2(c_0 - c_1)}}{\sqrt{c_3(c_0 - 1)}}$.
\end{lemma}

The proof can be found in the Lemma 4 and 5 in \cite{wang2015high} for elliptical distributions. The special case of Gaussian is also proved in the Lemma 3 of \cite{wang2015consistency}. Notice that the eigenvalue assumption in \cite{wang2015high} is not used for proving Lemma 4 and 5.

\begin{lemma} \label{lemma:tail}
Assume $x_i$ follows $EN(L, \Sigma)$. If $E[L^{-2}] < M_1$ for some constant $M_1 > 0$, $var(\epsilon) = \sigma^2$ and $\log p = o(n)$, then for any $0 < \alpha < 1$ we have
\begin{align*}
P\bigg(\|\eta\|_\infty\leq \frac{c_1\kappa^{-1}\tau^*}{6}\frac{n}{p}\bigg) \geq 1 - O\bigg(\frac{\sigma^2\kappa^4\log p}{\tau^{*2} n^{1-\alpha}}\bigg),
\end{align*}
where $\tau^*$ is defined as the minimum value for the important signals and $\kappa = cond(\Sigma)$.
\end{lemma}
To prove Lemma \ref{lemma:tail} we need the following two propositions.

\begin{Prop} (Lounici, 2008 \cite{lounici2008sup}; Nemirovski, 2000 \cite{akritastopics}) \label{prop:lounici}
Let $Y_i\in \mathcal{R}^p $ be random vectors with zero means and finite variances. Then we have for any $k$ norm with $k\in [2, \infty]$ and $p \geq 3$, we have
\begin{align}
E\big\|\sum_{i = 1}^n Y_i\big\|_k^2 \leq \tilde C\min\{k, \log p\}\sum_{i = 1}^n E\|Y_i\|_k^2,
\end{align}
where $\tilde C$ is some absolute constant.
\end{Prop}

As each row of $X$ can be represented as $X = \bar{L} Z\Sigma^{1/2}$, where $\bar{L} = diag(\sqrt{p}L_1/\|z_1\|_2, \cdots, \sqrt{p}L_n/\|z_n\|_2)$ and $Z$ is a matrix of independent Gaussian entries, i.e., $Z\sim N(0, I_p)$. For $Z$, we have the following result.
\begin{Prop}\label{lemma:Z}
Let $Z\sim N(0, I_p)$, then we have the minimum eigenvalue of $ZZ^T/p$ satisfies that
\begin{align*}
P\bigg(\lambda_{min}(ZZ^T/p) > (1 - \frac{n}{p} - \frac{t}{p})^2\bigg) \geq 1 - 2\exp(-t^2/2) 
\end{align*}
for any $t > 0$. Assume $p > c_0n$ for $c_0 > 1$ and take $t = \sqrt{n}$. When $n > 4c_0^2/(c_0 - 1)^2$, we have
\begin{align}
P\bigg(\lambda_{min}(ZZ^T/p) > c\bigg) \geq 1 - 2\exp(-n/2),\label{eq:half}
\end{align}
where $c = \frac{(c_0 - 1)^2}{4c_0^2}$.
\end{Prop}

The proof follows Corollary 5.35 in \cite{vershynin2010introduction}.

\begin{proof}[\textbf{Proof of Lemma \ref{lemma:tail}}]
Let $A = pX^T(XX^T)^{-1} \bar{L}$ and $Z = \bar{L}^{-1}X\Sigma^{-1/2}$. Then $\eta = p^{-1}A\bar{L}^{-1}\epsilon$.
\paragraph{Part 1. Bounding $|A_{ij}|$.} Consider the standard SVD on $Z$ as $Z = VDU^T$, where $V$ and $D$ are $n \times n$ matrices and $U$ is a $p\times n$ matrix. Because $Z$ is a matrix of iid Gaussian variables, its distribution is invariant under both left and right orthogonal transformation. In particular, for any $T\in\mathcal{O}(n)$, we have 
\begin{align*}
TVDU^T \stackrel{(d)}{=} VDU^T,
\end{align*}
i.e., $V$ is uniformly distributed on $\mathcal{O}(n)$ conditional on $U$ and $D$ (they are in fact independent, but we don't need such a strong condition).
Therefore, we have
\begin{align*}
A &= pX^T(XX^T)^{-1}L = p\Sigma^{\frac{1}{2}}Z^TL(LZ\Sigma Z^TL)^{-1}L = p\Sigma^{\frac{1}{2}}UDV^TL(LVDU^T\Sigma UDV^TL)^{-1}L\\
& = p\Sigma^{\frac{1}{2}}U(U^T\Sigma U)^{-1}D^{-1}V^T = \sqrt{p}\Sigma^{\frac{1}{2}}U(U^T\Sigma U)^{-1}\big(\frac{D}{\sqrt{p}}\big)^{-1}V^T.
\end{align*}
Because $V$ is uniformly distributed conditional on $U$ and $D$, the distribution of $A$ is also invariant under right orthogonal transformation conditional on $U$ and $D$, i.e., for any $T\in \mathcal{O}(n)$, we have
\begin{align}
A \stackrel{(d)}{=} AT. \label{eq:A}
\end{align}
Our first goal is to bound the magnitude of individual entries $A_{ij}$. Let $v_i = e_i^TAA^Te_i$, which is a function of $U$ and $D$ (see below). From \eqref{eq:A}, we know that $e_i^TA$ is uniformly distributed on the sphere $S^{n-1}(\sqrt{v_i})$ if conditional on $v_i$ (i.e., conditional on $U, D$), which implies that
\begin{align}
e_i^TA \stackrel{(d)}{=} \sqrt{v_i}\bigg(\frac{x_1}{\sqrt{\sum_{j=1}^n x_j^2}}, \frac{x_2}{\sqrt{\sum_{j=1}^n x_j^2}}, \cdots, \frac{x_n}{\sqrt{\sum_{j=1}^n x_j^2}}\bigg),\label{eq:A1}
\end{align}
where $x_j's$ are iid standard Gaussian variables. Thus, $A_{ij}$ can be bounded easily if we can bound $v_i$. Notice that for $v_i$ we have
\begin{align*}
v_i &= e_i^TAA^Te_i = p e_i^T\Sigma^{\frac{1}{2}}U(U^T\Sigma U)^{-1}\big(\frac{D^2}{p}\big)^{-1}(U^T\Sigma U)^{-1}U^T\Sigma^{\frac{1}{2}}e_i.\\
&= p e_i^TH (U^T\Sigma U)^{-\frac{1}{2}}\big(\frac{D^2}{p}\big)^{-1}(U^T\Sigma U)^{-\frac{1}{2}} H^Te_i\\
&\leq p e_i^THH^Te_i\cdot \lambda_{min}^{-1}(U^T\Sigma U)\cdot \lambda_{min}^{-1}\big(\frac{D^2}{p}\big)
\end{align*}
Here $H = \Sigma^{\frac{1}{2}}U(U^T\Sigma U)^{-1/2}$ is defined the same as in \cite{wang2015high} and can be bounded as $e_i^THH^Te_i \leq c_2n\kappa/p$ with probability $1 - 2\exp(-Cn)$ (see the proof of Lemma 3 in \cite{wang2015consistency}). Therefore, we have
\begin{align*}
P\bigg(v_i\leq c_2\kappa^2 \lambda_{min}^{-1}\big(\frac{D^2}{p}\big) n\bigg) \geq 1 - 2\exp(-Cn)
\end{align*}
Now applying the tail bound and the concentration inequality to \eqref{eq:A1} we have for any $t>0$ and any $C>0$
\begin{align}
P(|x_j|>t)\leq 2\exp(-t^2/2)\qquad P\bigg(\frac{\sum_{j=1}^n x_j^2}{n} \leq c_3\bigg) \leq \exp(-Cn). \label{eq:A_ineq}
\end{align}
Putting the pieces all together, we have for any $t>0$ and any $C>0$ that
\begin{align*}
P\bigg(\max_{ij} |A_{ij}| \leq \kappa t\sqrt{\frac{c_2}{c_3}}\lambda_{min}^{-\frac{1}{2}}\big(\frac{D^2}{p}\big)\bigg) \geq 1 - 2np\exp(-t^2/2) - 3p\exp(-Cn).
\end{align*}
Now according to \eqref{eq:half}, we can further bound $\lambda_{min}(D^2/p)$ and obtain that
\begin{align}
P\bigg(\max_{ij} |A_{ij}| \leq \sqrt{\frac{c_2}{cc_3}}\kappa t\bigg) \geq 1 - 2np\exp(-t^2/2) - 3p\exp(-Cn) - 2\exp(-n/2).\label{eq:Aij}
\end{align}

\paragraph{Part 2. Bounding $\eta$}
he second step is to use \eqref{eq:Aij} and Proposition \ref{prop:lounici} to bound $\eta$. The procedure follows similarly as in Lounici's paper. We first note that $\|z_i\|_2^2$ follows a chi-square distribution $\mathcal{X}^2(p)$. We have for any $t$
\begin{align*}
P\bigg(\frac{\|z_i\|_2^2}{p} \geq 1 + 2\sqrt{\frac{t}{p}} + \frac{2t}{p}\bigg)\leq e^{-t},
\end{align*}
from which we know
\begin{align}
P\bigg(\max_{i} p^{-1}\|z_i\|_2^2 < 5/2\bigg) \geq 1 - p e^{-p/4}. \label{eq:ep}
\end{align}
\par
Now define $W_j = (A_{1j}p^{-1/2}\|z_j\|_2L_j^{-1}\epsilon_j, A_{2j}p^{-1/2}\|z_j\|_2L_j^{-1}\epsilon_j,\cdots, A_{pj}p^{-1/2}\|z_j\|_2L_j^{-1}\epsilon_j$). It's clear that $\eta = \sum_{j=1}^n W_j/p$. Applying Proposition \ref{prop:lounici} to $W_j's$ with the $l_\infty$ norm and noticing tht $L_j$ is independent of $z_j$ we have
\begin{align*}
E\big\|\sum_{j = 1}^n W_j\big\|_\infty^2 \leq \log p\sum_{j = 1}^n E\|W_j\|_\infty^2\leq \log p \frac{7c_2}{cc_3}\sigma^2\kappa^2t^2 \sum_{j = 1}^n E[L_j^{-2}] \leq \frac{c_2}{cc_3}\sigma^2\kappa^2t^2M_1^2n\log p.
\end{align*}
Using the Markov inequality on $\eta$, we have for any $r > 0$
\begin{align*}
P\bigg(\|\eta\|_\infty\geq \frac{\sqrt{n}r}{p}\bigg) &= P\bigg(\frac{p}{\sqrt{n}}\|\eta\|_\infty\geq r\bigg)\leq \frac{p^2E\|\eta\|_\infty^2}{nr^2} =  \frac{E\|\sum_{j = 1}^n W_j\|_\infty^2}{nr^2}\\
&\leq \frac{7c_2\sigma^2\kappa^2 M_1^2 t^2 \log p}{cc_3r^2}.
\end{align*}
To match our previous result, we take $r = c_1\sqrt{n}\tau^*\kappa^{-1}/6$ and $t = n^{(1 - \alpha)/2}$ for some small $\alpha$,
\begin{align*}
P\bigg(\|\eta\|_\infty\leq \frac{c_1\kappa^{-1}\tau^*}{6}\frac{n}{p}\bigg)&\geq 1 - \frac{342 c_2\sigma^2\kappa^4 M_1}{c_1^2cc_3\tau^{*2}}\frac{\log p}{n^{\alpha}} - 2np\exp(-n^{1 - \alpha}/2) - 3p\exp(-Cn) - 2\exp(-n/2)\\
&\geq 1 - O\bigg(\frac{\sigma^2\kappa^4\log p}{\tau^{*2} n^{\alpha}}\bigg).
\end{align*}
\end{proof}

\begin{lemma} \label{lemma:unimportant}
Assume $var(Y) \leq M_0$. Define $\Phi = X^T(XX^T)^{-1}X$. If $p > c_0 n$ for some $c_0 > 1$, then we have for any $t > 0$
\begin{align*}
P\bigg(\max_i \sum_{j\neq i}|\Phi_{ij}\beta_j| \geq c_4\sqrt{M_0}\kappa^{\frac{3}{2}} t\frac{\sqrt{n}}{p}\bigg) \leq 2p e^{-t^2/2} + 5pe^{-Cn}.
\end{align*}
where $c_4, \kappa$ are defined in Lemma \ref{lemma:5}.
\end{lemma}

\begin{proof}
Following \cite{wang2015high,wang2015consistency}, we define $H = X^T(XX^T)^{-\frac{1}{2}}$. When $X\sim N(0, \Sigma)$, $H$ follows the $MACG(\Sigma)$ distribution as indicated in Lemma 3 in \cite{wang2015consistency} and Theorem 1 in \cite{wang2015high}. For simplicity, we only consider a particular case where $i = 1$. 
\par
For vector $v$ with $v_1 = 0$, we define $v' = (v_2, v_3, \cdots, v_p)^T$ and we can always identify a $(p - 1)\times (p - 1)$ orthogonal matrix $T'$ such that $T'v' = \|v'\|_2e_1'$ where $e_1'$ is a $(p - 1) \times 1$ unit vector with the first coordinate being 1. Now we define a new orthogonal matrix $T$ as
\begin{align*}
T = \begin{pmatrix}
1 & 0\\
0 & T'
\end{pmatrix}
\end{align*}
and we have
\begin{align*}
Tv = \begin{pmatrix}
1 & 0\\
0 & T'
\end{pmatrix}
\begin{pmatrix}
0 \\
v'
\end{pmatrix}
=
\begin{pmatrix}
0\\
\|v\|_2e_1'
\end{pmatrix}
= \|v\|_2 e_2.
\quad\mbox{and}\quad
e_1^T T^T = e_1^T\begin{pmatrix}
1 & 0\\
0 & T^{'T}
\end{pmatrix}
= e_1^T
\end{align*}
Therefore, we have
\begin{align*}
e_1^THH^Tv = e_1^TT^TTHH^TT^TTv = e_1^TT^THH^TT^Te_2 = \|v\|_2 e_1^T\tilde H\tilde H^Te_2.
\end{align*}
Since $H$ follows $MACG(\Sigma)$, $\tilde H = T^TH$ follows $MACG(T^T\Sigma T)$ for any fixed $T$. Therefore, we can apply Lemma \ref{lemma:5} again to obtain that
\begin{align*}
P\bigg(&|e_1^TX^T(XX^T)^{-1}Xv|\geq \|v\|_2 c_4\kappa t\frac{\sqrt{n}}{p}\bigg) = P\bigg(|e_1^THH^Tv|\geq \|v\|_2 c_4\kappa t\frac{\sqrt{n}}{p}\bigg)\\
& = P\bigg(\|v\|_2|e_1^T\tilde H\tilde H^Te_2|\geq \|v\|_2 c_4\kappa t\frac{\sqrt{n}}{p}\bigg) = P\bigg(\|v\|_2|\Phi_{12}|\geq \|v\|_2 c_4\kappa t\frac{\sqrt{n}}{p}\bigg)\\
&= P\bigg(|\Phi_{12}| \geq c_4\kappa t\frac{\sqrt{n}}{p}\bigg)\leq 5e^{-Cn} + 2e^{-t^2/2}.
\end{align*}
Applying the above result to $v = (0, \beta_*^{(-1)})$ we have
\begin{align*}
\sum_{j\neq 1} |\Phi_{1j}\beta_j| \leq c_4\kappa t \|\beta\|_2 \frac{\sqrt{n}}{p}
\end{align*}
with probability at least $1 - 5e^{-Cn} - 2e^{-t^2/2}$.
\par
In addition, we know that $var(Y) = \beta_*^T\Sigma\beta_* + \sigma^2 \leq M_0$ and thus
\begin{align*}
\|\beta\|_2 \leq \sqrt{M_0\kappa}.
\end{align*}
Consequently, we have
\begin{align*}
P\bigg(\max_i \sum_{j\neq i}|\Phi_{ij}\beta_j| \geq c_4\sqrt{M_0}\kappa^{\frac{3}{2}} t\frac{\sqrt{n}}{p}\bigg) \leq 2p e^{-t^2/2} + 5pe^{-Cn}.
\end{align*}
\end{proof}

Now we are ready to prove Theorem \ref{thm:1}
\begin{proof}[\textbf{Proof of Theorem \ref{thm:1}}]
Recall the definition of $\xi$ as $\xi = X^T(XX^T)^{-1}X\beta$. 
For any $i$ we have
\begin{align*}
\xi_i = e_i^TX^T(XX^T)^{-1}X\beta = \sum_{j\in S} \Phi_{ii}\beta_i +\sum_{j\neq i} \Phi_{ij}\beta_j,
\end{align*}
For the first term, we have
\begin{align*}
|\min_{ii}\beta_i| \geq c_1\kappa^{-1}\tau^*\frac{n}{p}\quad\forall i\in S^*
\end{align*}
with probability $1 - |S^*| e^{-Cn}$ and
\begin{align*}
|\min_{ii}\beta_i| \leq c_1\kappa\tau_*\frac{n}{p}\quad\forall i\in S_*
\end{align*}
with probability $1 - |S_*| e^{-Cn}$. Now, for the second term, using Lemma \ref{lemma:unimportant}, we have
\begin{align*}
\sum_{j\neq i} |\Phi_{ij}\beta_j| \leq \frac{c_1\kappa^{-1}\tau^*}{6}\quad \forall i = 1, 2, \cdots, p
\end{align*}
with probability at least $1 - 2p \exp\{-\frac{c_1^2\kappa^{-1}\tau^{*2}}{72c_4^2M_0}n\} - 5pe^{-Cn}$.
Therefore, we have for any $i\in S^*$
\begin{align*}
|\xi_i| \geq c_1\kappa^{-1}\tau^*\frac{n}{p} - \frac{c_1\kappa^{-1}\tau^*}{6}\frac{n}{p} \geq \frac{5c_1\kappa^{-1}\tau^*}{6}\frac{n}{p}.
\end{align*}
and for $i\in S_*$ we have
\begin{align*}
|\xi_i| \leq c_1\kappa\tau_*\frac{n}{p} + \frac{c_1\kappa^{-1}\tau^*}{6}\frac{n}{p} \leq \frac{7c_1\kappa^{-1}\tau^*}{12}\frac{n}{p},
\end{align*}
where we use the assumption that $\tau^* > 4\kappa^2\tau_*$. Now combining the result from Lemma \ref{lemma:tail}, we can obtain
\begin{align*}
P\bigg(\min_{i\in S^*} |\hat \beta_i| \geq \frac{2c_1\kappa^{-1}\tau^*}{3}\frac{n}{p}\bigg) \geq 1 - O\bigg(\frac{\sigma^2\kappa^4\log p}{\tau^{*2} n^{\alpha}}\bigg),
\end{align*}
and
\begin{align*}
P\bigg(\max_{i\in S_*} |\hat \beta_i| \leq \frac{7c_1\kappa^{-1}\tau^*}{12}\frac{n}{p}\bigg) \geq 1 - O\bigg(\frac{\sigma^2\kappa^4\log p}{\tau^{*2} n^{\alpha}}\bigg).
\end{align*}
Taking $\gamma = \frac{2c_1\kappa^{-1}\tau^*}{3}{n}{p}$, we have
\begin{align*}
P\bigg(\min_{i\in S^*}|\hat\beta_i| \geq \gamma \geq \max_{i\in S_*} |\hat\beta_i| \bigg) \geq 1 - O\bigg(\frac{\sigma^2\kappa^4\log p}{\tau^{*2} n^{\alpha}}\bigg).
\end{align*}
\end{proof}

\section*{Proof of Theorem \ref{thm:2} and \ref{thm:3}}
For the selected submodel $\mathcal{\hat M}_d$, we define $X_d$ to be the variables contained in $\mathcal{\hat M}_d$ and $X_{d, c}$ to be variables that are excluded from $\mathcal{\hat M}_d$. It is clear that
\begin{align*}
\hat\beta^{(OLS)}_d = (X_d^TX_d)^{-1}X_d^TY = \beta_d + (X_d^TX_d)^{-1}X_d^T\varepsilon + (X_d^TX_d)^{-1}X_d^TX_{d, c}\beta_{d, c} = \beta_d + \eta_d + \omega.
\end{align*}
To prove Theorem \ref{thm:2} is essentially to bound $\eta$ and $\omega$. Thus, we need following three lemmas.
\begin{lemma}[Garvesh, Wainwright and Yu. (2010) \cite{raskutti2010restricted}]\label{lemma:REC} Assume $Z\sim N(0, \Sigma)$.
There exists some absolute constant $c', c'' > 0$ such that 
\begin{align*}
\frac{\|Zv\|_2}{\sqrt{n}} \geq \frac{1}{4}\|\Sigma^{\frac{1}{2}}v\|_2 - 9\rho(\Sigma) \sqrt{\frac{\log p}{n}}\|v\|_1,\quad \forall v\in \mathcal{R}^p,
\end{align*}
with probability at least $1 - c''\exp(-c'n)$, where $\rho(\Sigma) = \max_{i = 1,2,\cdots,p} \Sigma_{ii}$.
\par
In our case, for any $v$ with $d$ nonzero coordinates, we have $\|v\|_1\leq \sqrt{d}\|v\|_2$, $\rho(\Sigma) = 1$ and $\|\Sigma^{1/2}v\|_2\geq \lambda^{\frac{1}{2}}_{\min}(\Sigma)\|v\|_2$. Therefore,
\begin{align*}
\frac{\|Zv\|_2}{\sqrt{n}} \geq \bigg(\frac{\lambda^{\frac{1}{2}}_{\min}(\Sigma)}{4} - 9 \sqrt{\frac{d\log p}{n}}\bigg)\|v\|_2,\quad  \|v\|_0\leq d.
\end{align*}
Thus, as long as $n \geq 6^4\kappa d\log p$, we have
\begin{align*}
\min_{|\mathcal{\hat M}| \leq d} \lambda_{min}^{1/2}(Z_{\mathcal{\hat M}}^TZ_{\mathcal{\hat M}}/n) \geq \frac{\lambda^{\frac{1}{2}}_{\min}(\Sigma)}{8}.
\end{align*}
\end{lemma}

\begin{lemma}\label{lemma:eta}
Assume $E[L^{-12}] \leq M_1$ and $e[L^{12}] \leq M_2$. For any $\mathcal{\hat M}$ such that $S^*\subset \mathcal{\hat M}$ and $|\mathcal{\hat M}| \leq d$, we have for any $\alpha > 0$
\begin{align*}
P\bigg(\max_{|\mathcal{\hat M}|\leq d} \|\eta_d\|_\infty \leq \sigma \sqrt{\frac{\log p}{n^{\alpha}}}\bigg) = 1 - O\bigg(\frac{\lambda_*^{-2} d\log d}{n^{\frac{1}{3}(1 - \alpha)}} + \frac{M_1 + M_2}{n^{\frac{1}{3}(1 - 4\alpha)}}\bigg),
\end{align*}
where $\lambda_* = \lambda_{\min}(\Sigma)$.
\end{lemma}

\begin{proof}
Define $A = (X_d^TX_d)^{-1}X_d^T$, we have
\begin{align*}
\eta = (X_d^TX_d)^{-1}X_d^T\epsilon = A\epsilon.
\end{align*}
For $A$, we can bound its entries as
\begin{align*}
\max_{ij}|A_{ij}| &\leq \max_{ij} |e_i^T(X_d^TX_d)^{-1}X_d^Te_j| \leq \max_{ij} \|e_i^T(X_d^TX_d)^{-1}\|_1 \|X_d^Te_j\|_\infty\\
&\leq \sqrt{d}\max_{ij}\|e_i^T(X_d^TX_d)^{-1}\|_2 \max_{ij}|X_d^T| \leq \frac{\sqrt{d}}{n}\lambda^{-1}_{min}\bigg(\frac{X_d^TX_d}{n}\bigg)\max_{ij}|X_d^T|.
\end{align*}
Recall that $X = \bar{L}Z\Sigma^{1/2}$, where $\bar{L} = diag(\sqrt{p}L_1/\|z_1\|_2, \cdots, \sqrt{p}L_n/\|z_n\|_2)$ and thus $X_d$ possesses a representation as $X_d = \bar{L}Z\Sigma^{1/2}_d$, where $\Sigma^{1/2}_d$ is an $p \times d$ matrix formed by the selected $d$ columns of $\Sigma^{1/2}$. We can now further bound $\lambda^{-1}_{min}\bigg(\frac{X_d^TX_d}{n}\bigg)$ as
\begin{align*}
\lambda^{-1}_{min}\bigg(\frac{X_d^TX_d}{n}\bigg) &= \lambda^{-1}_{min}\bigg(\frac{\Sigma_d^{\frac{T}{2}}Z^T\bar{L}^T\bar{L}Z\Sigma_d^{\frac{1}{2}}}{n}\bigg)\\
&\leq \bigg(\lambda_{\min}(\bar{L}^T\bar{L})\lambda_{\min}(\Sigma_d^{\frac{T}{2}}Z^TZ\Sigma_d^{\frac{1}{2}}/n)\bigg)^{-1}.
\end{align*}
Using Lemma \ref{lemma:REC}, it is clear that
\begin{align*}
\min_{|\mathcal{\hat M}|\leq d} \lambda_{\min}(\Sigma_d^{\frac{T}{2}}Z^TZ\Sigma_d^{\frac{1}{2}}/n) \geq \frac{\lambda_{\min}(\Sigma)}{64}\geq \frac{\lambda_*}{64},
\end{align*}
with probability at least $1 - O(e^{-c'n})$. In addition, since $E[L^{-12}] \leq M_1$ and $E[L^{12}] \leq M_2$, we have for any $k_1 > 0, k_2 > 0$
\begin{align*}
P(L^2 \leq k_1 ) \leq k_1^6 M_1\quad\mbox{and}\quad P(L \geq k_2) \leq \frac{M_2}{k_2^{12}}.
\end{align*}
Combining with equation \eqref{eq:ep} implies that
\begin{align*}
\lambda_{\min}(\bar{L}^T\bar{L}) \geq \frac{2 k_1}{5},
\end{align*}
with probability at least $1 - pe^{-p/4} - nk_1^6M_1$. Therefore, we have
\begin{align*}
\max_{|\mathcal{\hat M}|\leq d} \lambda^{-1}_{min}\bigg(\frac{X_d^TX_d}{n}\bigg) \leq \frac{162}{\lambda_* k_1}.
\end{align*}
with probability $1 - O(nk_1^6M_1)$.
\par
For $\max_{ij}|X_d^T|$, we just need to bound $\max_{ij} X_{ij}$. Using the representation $X = \bar{L}Z\Sigma^{1/2}$, we know that
\begin{align*}
X_{ij} = \frac{\sqrt{p}L_i}{\|z_i\|_2} Z_i\Sigma^{1/2} e_j.
\end{align*}
It is easy to see that $Z_i\Sigma^{1/2} e_j$ is a Gaussian random variable with mean zero and variance 1, thus for any $t > 0$
\begin{align*}
P(|Z_i\Sigma^{1/2} e_j| \geq t) \leq 2 e^{-t^2/2}.
\end{align*}
In addition, $\|z_i\|_2^2/p$ follows a $\mathcal{X}^2(p)$ and we have
\begin{align*}
P\bigg(\frac{\|z_i\|_2^2}{p} \geq 1 - 2\sqrt{\frac{t}{p}}\bigg) \geq 1 - e^{-t}.
\end{align*}
Taking $t = p/4$, we have $\max_{i} \|z_i\|_2/\sqrt{p} \geq 1/2$ with probability at least $1 - ne^{-p/4}$ and thus
\begin{align*}
P(\max_{ij}|X_{ij}| \leq 4k_2\sqrt{\log p}) \geq 1 - \frac{M_2 n}{k_2^{12}} - 2p^{-1} - ne^{-p/4}.
\end{align*}
Combining all pieces of results, we obtain that
\begin{align*}
P\bigg(\min_{|\mathcal{\hat M}|\leq d} \max_{ij} |A_{ij}| \leq \frac{648 k_2 \sqrt{d} \sqrt{\log p}}{\lambda_* k_1 n}\bigg) \geq 1 - O\bigg(nk_1^6M_1 + \frac{nM_2}{k_2^{12}}\bigg).
\end{align*}
Following a similar argument in proving Lemma \ref{lemma:tail}, we define $W_j = (A_{1j}\epsilon_j, A_{2j}\epsilon_j, \cdots, A_{dj}\epsilon_j)$ and then
\begin{align*}
\eta = \sum_{j = 1}^n W_j.
\end{align*}
Using Proposition \ref{prop:lounici}, we have
\begin{align*}
E\|\eta\|_\infty^2 = E\|\sum_{j = 1}^n W_j\|_\infty^2 \leq \tilde C\log d\sum_{j = 1}^n E\|W_j\|_\infty^2 \leq O\bigg(\frac{\sigma^2 k_2^2}{\lambda_*^2k_1^2}\frac{d\log d\log p}{n}\bigg).
\end{align*}
Using the Markov inequality implies that for any $r > 0$
\begin{align*}
P\bigg(\max_{|\mathcal{\hat M}|\leq d}\|\eta\|_\infty > r\bigg)\leq \frac{\|\eta\|_\infty^2}{r^2} = O\bigg(\frac{\sigma^2 k_2^2}{\lambda_*^2k_1^2r^2}\frac{d\log d\log p}{n}\bigg) + O\bigg(nk_1^6M_1 + \frac{nM_2}{k_2^{12}}\bigg).
\end{align*}
Let $r = \sigma \sqrt{\frac{\log p}{n^{\alpha}}}$, $k_1 = n^{-\frac{2(1 - \alpha)}{9}}$ and $k_2 = n^{\frac{1 - \alpha}{9}}$, we have
\begin{align*}
P\bigg(\max_{|\mathcal{\hat M}|\leq d}\|\eta\|_\infty \leq \sigma \sqrt{\frac{\log p}{n^{\alpha}}}\bigg) = 1 - O\bigg(\frac{\lambda_*^{-2} d\log d}{n^{\frac{1}{3}(1 - \alpha)}} + \frac{M_1 + M_2}{n^{\frac{1}{3}(1 - 4\alpha)}}\bigg)
\end{align*}
\end{proof}

\begin{lemma}\label{lemma:omega}
Assume $E[L^{-12}] \leq M_1$ and $e[L^{12}] \leq M_2$. For any $\mathcal{\hat M}$ such that $S^*\subset \mathcal{\hat M}$ and $|\mathcal{\hat M}| \leq d$. Assume that $d - |S^*|\leq \tilde c$ and $\sum_{i \not\in S^*} |\beta_i|^\iota \leq R$ for some $\iota\in (0, 1)$, then for any $\alpha > 0$, we have
\begin{align*}
P\bigg(\max_{|\mathcal{\hat M}|\leq d} \|w\|_2 \leq \sigma\sqrt{\frac{\log p}{n^\alpha}}\bigg) \geq 1 - O\bigg(\frac{(M_1 + M_2)R^3}{(\log p)^{2\iota}n^{3 - 4\alpha - 2\iota}}\bigg).
\end{align*}
\end{lemma}
\begin{proof}
According to our definition that $\omega = (X_d^TX_d)^{-1}X_d^TX_{d, c}\beta_{d, c}$, we can directly bound the $l_2$ norm of $\omega$ as
\begin{align*}
\|\omega\|_2^2 = \beta_{d,c}^TX^T_{d, c}X_d(X_d^TX_d)^{-2}X_d^TX_{d, c}\beta_{d, c}\leq \frac{1}{n}\beta_{d,c}^TX^T_{d, c}X_{d, c}\beta_{d, c}\lambda^{-1}_{min}\bigg(\frac{X_d^TX_d}{n}\bigg)
\end{align*}
where $\lambda^{-1}_{min}\bigg(\frac{X_d^TX_d}{n}\bigg)$ has already obtained a bound in Lemma \ref{lemma:eta} as
\begin{align*}
\max_{|\mathcal{\hat M}|\leq d} \lambda^{-1}_{min}\bigg(\frac{X_d^TX_d}{n}\bigg) \leq \frac{162}{\lambda_* k_1}.
\end{align*}
with probability $1 - O(nk_1^6M_1)$. Now for $\frac{1}{n}\beta_{d,c}^TX^T_{d, c}X_{d, c}\beta_{d, c}$ we have
\begin{align*}
\frac{1}{n}\beta_{d,c}^TX^T_{d, c}X_{d, c}\beta_{d, c} = \frac{1}{n}\beta_{d,c}^T\Sigma_{d,c}^{T/2}Z^T\bar{L}^T\bar{L}Z\Sigma_{d,c}^{1/2}\beta_{d,c} \leq \frac{1}{n}\beta_{d,c}^T\Sigma_{d,c}^{T/2}Z^TZ\Sigma_{d,c}^{1/2}\beta_{d,c} \max_{i}\frac{pL_i^2}{\|z_i\|_2^2}
\end{align*}
Since $Z\sim N(0, I_p)$, we can choose an orthogonal matrix $Q$ such that $\beta_{d,c}\Sigma_{d,c}^{1/2} = e_1Q\|\beta_{d, c}\Sigma_{d, c}^{1/2}\|_2$ and 
\begin{align*}
\frac{1}{n}\beta_{d,c}^T\Sigma_{d,c}^{T/2}Z^TZ\Sigma_{d,c}^{1/2}\beta_{d,c} = \|\beta_{d, c}\Sigma_{d, c}^{1/2}\|_2^2 e_1\tilde Z^T\tilde Z e_1^T \leq \|\beta_{d, c}\|_2^2\lambda^* e_1\tilde Z^T\tilde Ze_1,
\end{align*}
where $\tilde Z\sim N(0, I_p)$. It is easy to see that for any $t > 0$
\begin{align*}
P\bigg(\frac{e_1^T\tilde Z^T\tilde Ze_1}{n} \leq 1 + 2\sqrt{\frac{t}{n}} + \frac{2t}{n}\bigg) \geq 1 - e^{-t}.
\end{align*}
and $\|\beta_{d,c}\|_2^2 \leq \tau_*^{2 - \iota}R$. Thus, taking $t = (1 + \tilde c)\log p$, we have
\begin{align*}
\max_{|\mathcal{\hat M}|\leq d}\frac{1}{n}\beta_{d,c}^T\Sigma_{d,c}^{T/2}Z^TZ\Sigma_{d,c}^{1/2}\beta_{d,c}\leq 5\tau_*^{2 - \iota}R\lambda^*
\end{align*}
with probability $1 - p^{-1}$ as long as $n \geq (1 + \tilde c)\log p$ where $\tilde c$ is the upper bound on $d - |S^*|$. For $\max_i pL_i^2/\|z_i\|_2^2$, we follow the same argument in Lemma \ref{lemma:eta}
\begin{align*}
P\bigg(\max_i \frac{pL_i^2}{\|z_i\|_2^2} \leq 2k_2^2\bigg) \geq 1 - ne^{-p/4} - \frac{nM_2}{k_2^{12}}.
\end{align*}
Putting all pieces together, we have
\begin{align*}
\max_{|\mathcal{\hat M}|\leq d} \|w\|_2 \leq 36\tau_*^{1 - \frac{\iota}{2}}R^{\frac{1}{2}}\kappa^{\frac{1}{2}}\sqrt{\frac{k_2^2}{k_1}},
\end{align*}
with probability at least $1 - O\bigg(\frac{nM_2}{k_2^{12}} +  nk_1^6M_1\bigg)$. According to our assumption that $\tau_*\leq \frac{\sigma}{\kappa}\sqrt{\frac{\log p}{n}}$ and taking $k_1 = \frac{n^{\iota/4}R^{1/2}}{(\log p)^{\iota/4}n^{(1 - \alpha)/2}}$ and $k_2 = 1/\sqrt{k_1}$ we have
\begin{align*}
P\bigg(\max_{|\mathcal{\hat M}|\leq d} \|w\|_2 \leq \sigma\sqrt{\frac{\log p}{n^\alpha}}\bigg) \geq 1 - O\bigg(\frac{(M_1 + M_2)R^3}{(\log p)^{2\iota}n^{3 - 4\alpha - 2\iota}}\bigg).
\end{align*}
\end{proof}

We are now ready to prove Theorem \ref{thm:2}
\begin{proof}[\textbf{Proof of Theorem \ref{thm:2}}]
We just need to combine the results of Lemma \ref{lemma:eta} and \ref{lemma:omega}, i.e., 
\begin{align*}
\hat\beta^{(OLS)}_d = \beta_d + \eta + \omega,
\end{align*}
where
\begin{align*}
P\bigg(\max_{|\mathcal{\hat M}|\leq d}\|\eta\|_\infty \leq \sigma \sqrt{\frac{\log p}{n^{\alpha}}}\bigg) = 1 - O\bigg(\frac{\lambda_*^{-2} d\log d}{n^{\frac{1}{3}(1 - \alpha)}} + \frac{M_1 + M_2}{n^{\frac{1}{3}(1 - 4\alpha)}}\bigg)
\end{align*}
and 
\begin{align*}
P\bigg(\max_{|\mathcal{\hat M}|\leq d} \|w\|_2 \leq \sigma\sqrt{\frac{\log p}{n^\alpha}}\bigg) \geq 1 - O\bigg(\frac{(M_1 + M_2)R^3}{(\log p)^{2\iota}n^{3 - 4\alpha - 2\iota}}\bigg).
\end{align*}
Therefore, we have
\begin{align*}
P\bigg(\max_{|\mathcal{\hat M}| \leq d, S^*\subset\mathcal{\hat M}}\|\hat\beta^{(OLS)}_d - \beta_d\|_\infty\leq 2\sigma\sqrt{\frac{\log p}{n^\alpha}}\bigg) = 1 - O\bigg(\frac{\lambda_*^{-2} d\log d}{n^{\frac{1}{3}(1 - \alpha)}} + \frac{M_1 + M_2}{n^{\frac{1}{3}(1 - 4\alpha)}} + \frac{(M_1 + M_2)R^3}{(\log p)^{2\iota}n^{3 - 4\alpha - 2\iota}}\bigg)
\end{align*}
\end{proof}

\begin{proof}[\textbf{Proof of Theorem \ref{thm:3}}]
Recall that $X_d$ consists of variables contained in $\mathcal{\hat M}_d$, the definition of $\hat \beta(r)^{(Ridge)}$ becomes 
\begin{align*}
\hat \beta(r)^{(Ridge)} &= (X_d^TX_d + rI_d)^{-1}X_d^TX_d\beta + (X_d^TX_d + rI_d)^{-1}X_d^T\varepsilon + (X_d^TX_d + rI_d)^{-1}X_d^TX_{d, c}\beta_{d,c}\\
&  = \beta - r(X_d^TX_d + rI_d)^{-1}\beta + (X_d^TX_d + rI_d)^{-1}X_d^T\varepsilon + (X_d^TX_d + rI_d)^{-1}X_d^TX_{d, c}\beta_{d,c}\\
& = \beta - \tilde \xi(r) + \tilde \eta(r) + \tilde \omega(r).
\end{align*}
For $\tilde \xi(r)$ we have 
\begin{align*}
\|\tilde \xi(r)\|_2^2\leq r^2\beta^T(X_d^TX_d + rI_d)^{-2}\beta \leq \frac{r^2 \|\beta\|_2^2}{n^2\lambda_{min}^2(X_d^TX_d/n + r/n)} \leq \frac{8^4 r^2 \kappa^3 M_0}{n^2} 
\end{align*}
As proved in Lemma \ref{lemma:eta}, we know that
\begin{align*}
\max_{|\mathcal{\hat M}|\leq d} \lambda_{min}\bigg(\frac{X_d^TX_d}{n}\bigg) \geq \frac{\lambda_* k_1}{162}.
\end{align*}
with probability $1 - O(nk_1^6M_1)$. 
Adding $r/n$ to the above matrix will only increase the smallest eigenvalue. Thus, we have
\begin{align*}
\|\tilde \xi(r)\|_2\leq r^2\beta^T(X_d^TX_d + rI_d)^{-2}\beta \leq \frac{162 r\lambda^* M_0}{n\lambda_* k_1} = \frac{162 r\kappa M_0}{n k_1}.
\end{align*}
Where we used $M_0 \geq var(Y) \geq \|\beta\|_2^2\lambda_{max}^{-1}(\Sigma)$. Choosing $k_1 = n^{-\frac{2(1 - \alpha)}{9}}$, we have
\begin{align*}
P\bigg(\max_{|\mathcal{\hat M}|\leq d} \|\tilde \xi(r)\|_2 \leq \frac{162 r\kappa M_0}{n^{\frac{1}{9}(7 + 2\alpha)}}\bigg) = 1 - O\bigg(\frac{M_1}{n^{\frac{1}{3}(1 - 4\alpha)}}\bigg),
\end{align*}
which implies that as long as $r \leq \frac{\sigma n^{(7/9 - 5\alpha/18)}\sqrt{\log p}}{162\kappa M_0}$, we have
\begin{align*}
P\bigg(\max_{|\mathcal{\hat M}|\leq d} \|\tilde \xi(r)\|_2 \leq \sigma\sqrt{\frac{\log p}{n^\alpha}}\bigg) = 1 - O\bigg(\frac{M_1}{n^{\frac{1}{3}(1 - 4\alpha)}}\bigg).
\end{align*}
In addition, the proof for $\|\eta\|_\infty$ and $\|\omega\|_2$ shows that the only key quantity that has changed is $\max_{|\mathcal{\hat M}|\leq d} \lambda_{min}\bigg(\frac{X_d^TX_d}{n}\bigg)$ which is replaced by $\max_{|\mathcal{\hat M}|\leq d} \lambda_{min}\bigg(\frac{X_d^TX_d + rI_d}{n}\bigg)$ for $\beta^{(ridge)}$. While the latter is trivially lower bounded by the former, we thus have

\begin{align*}
P\bigg(\max_{|\mathcal{\hat M}|\leq d}\|\tilde \eta(r)\|_\infty \leq \sigma \sqrt{\frac{\log p}{n^{\alpha}}}\bigg) = 1 - O\bigg(\frac{\lambda_*^{-2} d\log d}{n^{\frac{1}{3}(1 - \alpha)}} + \frac{M_1 + M_2}{n^{\frac{1}{3}(1 - 4\alpha)}}\bigg)
\end{align*}
and 
\begin{align*}
P\bigg(\max_{|\mathcal{\hat M}|\leq d} \|\tilde w(r)\|_2 \leq \sigma\sqrt{\frac{\log p}{n^\alpha}}\bigg) \geq 1 - O\bigg(\frac{(M_1 + M_2)R^3}{(\log p)^{2\iota}n^{3 - 4\alpha - 2\iota}}\bigg).
\end{align*}
Consequently, we have
\begin{align*}
P\bigg(\max_{|\mathcal{\hat M}| \leq d, S^*\subset\mathcal{\hat M}}\|\hat\beta^{(ridge)}_d - \beta_d\|_\infty\leq 3\sigma\sqrt{\frac{\log p}{n^\alpha}}\bigg) = 1 - O\bigg(\frac{\lambda_*^{-2} d\log d}{n^{\frac{1}{3}(1 - \alpha)}} + \frac{2M_1 + M_2}{n^{\frac{1}{3}(1 - 4\alpha)}} + \frac{(M_1 + M_2)R^3}{(\log p)^{2\iota}n^{3 - 4\alpha - 2\iota}}\bigg),
\end{align*}
as long as
\begin{align*}
r \leq \frac{\sigma n^{(7/9 - 5\alpha/18)}\sqrt{\log p}}{162\kappa M_0}.
\end{align*}
\end{proof}

\begin{proof}[\textbf{Proof of Corollary \ref{cor:1}}]
As mentioned before, we have $\hat\beta^{(OLS)} = \beta_{\mathcal{\tilde M}_d} + (X_{\mathcal{\tilde M}_d}^TX_{\mathcal{\tilde M}_d})^{-1}X_{\mathcal{\tilde M}_d}\varepsilon$. Because $\varepsilon_i\sim N(0,\sigma^2)$ for $i = 1, 2, \cdots, n$, we have for any $i\in \mathcal{\tilde M}_d$,
\begin{align}
  \tilde\eta_i = e_i^T(X_{\mathcal{\tilde M}_d}^TX_{\mathcal{\tilde M}_d})^{-1}X_{\mathcal{\tilde M}_d}^T\varepsilon\sim N(0, \sigma^2e_i^T(X_{\mathcal{\tilde M}_d}^TX_{\mathcal{\tilde M}_d})^{-1}e_i)\stackrel{(d)}{=} \sigma\sqrt{e_i^T(X_{\mathcal{\tilde M}_d}^TX_{\mathcal{\tilde M}_d})^{-1}e_i} N(0,1).
\end{align}
Likewise in the proof of Lemma \ref{lemma:tail2}, we know that as long as $n \geq 64\kappa d\log p$
\begin{align*}
\lambda_{min}(X_{\mathcal{\tilde M}_d}^TX_{\mathcal{\tilde M}_d}/n)\geq \frac{1}{64\kappa}.
\end{align*}
Thus, we have

$$\max_{i\in \mathcal{\tilde M}_d} e_i^T(X_{\mathcal{\tilde M}_d}^TX_{\mathcal{\tilde M}_d})^{-1}e_i \leq 64\kappa/n.$$ 

Therefore, for any $t > 0$ and $i\in \mathcal{\tilde M}_d$, with probability at least $1 - c''\exp(-c'n) - 2\exp(-t^2/2)$ we have
\begin{align*}
|\tilde \eta_i| \leq \sigma t \sqrt{e_i^T(X_{\mathcal{\tilde M}_d}^TX_{\mathcal{\tilde M}_d})^{-1}e_i} \leq \frac{8\kappa^{\frac{1}{2}}\sigma t}{\sqrt n}.
\end{align*}
Then for any $\delta > 0$, if $n > \log(2c''/\delta)/c'$, then with probability at least $1 - \delta$ we have
\begin{align}
\max_{i\in\mathcal{\tilde M}_d} |\tilde \eta_i| \leq 8\sigma \sqrt{\frac{2\kappa\log(4d/\delta)}{n}}.\label{eq:ols1}
\end{align}
Because $\sigma$ needs to estimated from the data, we need to obtain a bound as well. Notice that $\hat\sigma^2$ is an unbiased estimator for $\sigma$, and
\begin{align*}
  \hat\sigma^2= \sigma^2 \epsilon^T(I_n - X_{\mathcal{\tilde M}_d}(X_{\mathcal{\tilde M}_d}^TX_{\mathcal{\tilde M}_d})^{-1}X_{\mathcal{\tilde M}_d})\epsilon \sim \frac{\sigma^2\mathcal{X}^2(n-d)}{n - d},
\end{align*}
where $\mathcal{X}^2(k)$ denotes a chi-square random variable with degree of freedom $k$. Using Proposition 5.16 in \cite{vershynin2010introduction}, we can bound $\hat\sigma^2$ as follows. Let $K = \|\mathcal{X}^2(1) - 1\|_{\psi_1}$. There exists some $c_5>0$ such that for any $t \geq 0$ we have,
\begin{align*}
  P\bigg(\bigg|\frac{\mathcal{X}^2(n-d)}{n - d} - 1\bigg|\geq t\bigg)\leq 2\exp\bigg\{-c_5\min\bigg(\frac{t^2(n-d)}{K^2}, \frac{t(n-d)}{K}\bigg)\bigg\}.
\end{align*}
Hence for any $\delta>0$, if $n > d + 4K^2\log(2/\delta)/c_5$, then with probability at least $1-\delta$ we have,
\begin{align*}
  |\hat\sigma^2 - \sigma^2| \leq \sigma^2/2,
\end{align*}
which implies that
\begin{align*}
  \frac{1}{2}\sigma^2\leq \hat\sigma^2\leq \frac{3}{2}\sigma^2.
\end{align*}
Then we know that
\begin{align*}
\max_{i\in\mathcal{\tilde M}_d} |\tilde \eta_i|\leq 8\sigma \sqrt{\frac{2\kappa\log(4d/\delta)}{n}} \leq 8\sqrt{2}\hat \sigma \sqrt{\frac{2\kappa\log(4d/\delta)}{n}} \leq 8\sqrt{3}\sigma \sqrt{\frac{2\kappa\log(4d/\delta)}{n}}.
\end{align*}
Now define $\gamma' = 8\sqrt{2}\hat \sigma \sqrt{\frac{2\kappa\log(4d/\delta)}{n}}$. If the signal $\tau = \min_{i\in S}|\beta_i|$ satisfies that
\begin{align*}
\tau \geq 24\sigma \sqrt{\frac{2\kappa\log(4d/\delta)}{n}},
\end{align*}
then with probability at least $1 - 2 \delta$, for any $i\not\in S$
\begin{align*}
|\hat\beta_i| = |\tilde \eta_i| \leq 8\sigma \sqrt{\frac{2\kappa\log(4d/\delta)}{n}} \leq \gamma',
\end{align*}
and for $i\in S$ we have
\begin{align*}
|\hat\beta_i| \geq \tau - \max_{i\in\mathcal{\tilde M}_d}|\tilde \eta_i| \geq 16\sigma \sqrt{\frac{2\kappa\log(4d/\delta)}{n}} \geq \gamma'.
\end{align*}
\end{proof}

\section*{Proof of Theorem 4}
The result of Theorem 4 can be immediately implied from Theorem \ref{thm:1}, \ref{thm:2}, \ref{thm:3}.

\end{document}